\documentclass[a4paper]{article}
\usepackage[utf8]{inputenc}
\usepackage{amsthm}
\usepackage{amsmath,amssymb}
\usepackage{mathtools}
\usepackage{enumitem}
\usepackage{tikz}
\usetikzlibrary{decorations.pathreplacing}

\usepackage{listings}
\usepackage{hyperref}

\usetikzlibrary{automata}

\newtheorem{definition}{Definition}
\newtheorem{theorem}[definition]{Theorem}

\newtheorem{lemma}[definition]{Lemma}
\newtheorem{example}[definition]{Example}
\newtheorem{proposition}[definition]{Proposition}

\newenvironment{roster}
{\begin{enumerate}[font=\upshape,label=(\alph*)]}
  {\end{enumerate}}

\newcommand{\N}{\mathbb{N}}
\newcommand{\R}{\mathbb{R}}
\newcommand{\lcm}{\mathrm{lcm}}
\newcommand{\GenSubset}{\mathrm{GenSubset}}
\newcommand{\GenState}{\mathrm{GenState}}
\newcommand{\Testing}{\mathrm{Testing}}
\newcommand{\gst}[1]{\mathsf{#1}}

\title{Winning Sets of Regular Languages: Descriptional and Computational Complexity}

\author{Pierre Marcus \\
  M2 informatique fondamentale \\
  \'Ecole Normale Supérieure de Lyon, Lyon, France \\
  \texttt{pierre.marcus@ens-lyon.fr}
  \and
  Ilkka Törmä\footnote{Author supported by Academy of Finland grant 295095.} \\
  Department of Mathematics and Statistics \\
  University of Turku, Turku, Finland \\
  \texttt{iatorm@utu.fi}}

\begin{document}
\maketitle

\begin{abstract}
  We investigate certain word-construction games with variable turn orders.
  In these games, Alice and Bob take turns on choosing consecutive letters of a word of fixed length, with Alice winning if the result lies in a predetermined target language.
  The turn orders that result in a win for Alice form a binary language that is regular whenever the target language is, and we prove some upper and lower bounds for its state complexity based on that of the target language.
  We also consider the computational complexity of membership and intersection problems of winning sets.

  Kaywrods:
  state complexity, regular languages, winning sets
\end{abstract}

\section{Introduction}
\label{sec:Intro}

Let us define a word-construction game of two players, Alice and Bob, as follows.
Choose a set of binary words $L \subseteq \{0, 1\}^*$ called the \emph{target set}, a length $n \geq 0$ and a word $w \in \{A, B\}^n$ called the \emph{turn order}, where $A$ stands for Alice and $B$ for Bob.
The players construct a word $v \in \{0, 1\}^n$ so that, for each $i = 0, 1, \ldots, n-1$ in this order, the player specified by $w_i$ chooses the symbol $v_i$.
If $v \in L$, then Alice wins the game, and otherwise Bob wins.
The existence of a winning strategy for Alice depends on both the target set and the turn order.
We fix the target set $L$ and define its \emph{winning set} $W(L)$ as the set of those words over $\{A,B\}$ that result in Alice having a winning strategy.

Winning sets were defined under this name in~\cite{salo2014playing} in the context of symbolic dynamics, but they have been studied before that under the name of \emph{order-shattering sets} in~\cite{Anstee2002,Friedl2003}.
The winning set has several interesting properties: it is downward closed in the index-wise partial order induced by $A < B$ (as changing $B$ to $A$ always makes the game easier for Alice) and it preserves the number of words of each length.
This latter property was used in~\cite{peltomaki2019winning} to study the growth rates of substitutive subshifts.

If the language $L$ is regular, then so is $W(L)$, as it can be recognized by an alternating finite automaton (AFA) \cite{salo2014playing}, which only recognizes regular languages~\cite{Chandra1981}.
Thus we can view $W$ as an operation on the class of binary regular languages, and in this article we study its state complexity in the general case and in several subclasses.
In our construction the AFA has the same state set as the original DFA, so our setting resembles parity games, where two players construct a path in a finite automaton~\cite{Zielonka1998}.
The main difference is that in a parity game, the player who chooses the next move is the owner of the current state, whereas here it is determined by the turn order word.

In the general case, the size of the minimal DFA for $W(L)$ can be doubly exponential in that of $L$.
We derive a lower, but still superexponential, upper bound for bounded regular languages (languages that satisfy $L \subseteq w_1^* w_2^* \cdots w_k^*$ for some words $w_i$).
We also study certain bounded permutation invariant languages, where membership is defined only by the number of occurrences of each symbol.
In particular, we explicitly determine the winning sets of the languages $L_k = (0^* 1)^k 0^*$ of words with exactly $k$ occurrences of $1$.
In the final section we consider the computational complexity of determining membership in $W(L)$ and whether a given regular language intersects it, which are P-complete and PSPACE-complete respectively when regular languages are given as finite automata.

In this article we only consider the binary alphabet, but we note that the definition of the winning set can be extended to languages $L \subseteq \Sigma^*$ over an arbitrary finite alphabet $\Sigma$ in a way that preserves the properties of downward closedness and $|L| = |W(L)|$.
The turn order word is replaced by a word $w \in \{1, \ldots, |\Sigma|\}^*$.
On turn $i$, Alice chooses a subset of size $w_i$ of $\Sigma$, and Bob chooses the letter $v_i$ from this set.

This article is an extended version of the conference publication~\cite{confVersion}.

\section{Definitions}

We present the standard definitions and notations used in this article.
For a set $\Sigma$, we denote by $\Sigma^*$ the set of finite words over it, and the length of a word $w \in \Sigma^n$ is $|w| = n$.
The notation $|w|_a$ means the number of occurrences of symbol $a \in \Sigma$ in $w$.
The empty word is denoted by $\lambda$.
For a language $L \subseteq \Sigma^*$ and $w \in \Sigma^*$, denote $w^{-1} L = \{ v \in \Sigma^* \;|\; w v \in L \}$.
We say $L$ is \emph{(word-)bounded} if $L \subseteq w_1^* \cdots w_k^*$ for some words $w_1, \ldots, w_k \in \Sigma^*$.
Bounded languages have been studied from the state complexity point of view in \cite{boundedRegular}.

A finite state automaton is a tuple $\mathcal{A} = (Q, \Sigma, q_0, \delta, F)$ where $Q$ is a finite state set, $\Sigma$ a finite alphabet, $q_0 \in Q$ the initial state, $\delta$ is the transition function and $F \subseteq Q$ is the set of final states.
The language accepted from state $q \in Q$ is denoted $\mathcal{L}_q(\mathcal{A}) \subseteq \Sigma^*$, and the language of $\mathcal{A}$ is $\mathcal{L}(\mathcal{A}) = \mathcal{L}_{q_0}(\mathcal{A})$.
The type of $\delta$ and the definition of $\mathcal{L}(\mathcal{A})$ depend on which kind of automaton $\mathcal{A}$ is.
\begin{itemize}
\item
  If $\mathcal{A}$ is a deterministic finite automaton, or DFA, then $\delta : Q \times \Sigma \to Q$ gives the next state from the current state and an input symbol.
  We extend it to $Q \times \Sigma^*$ by $\delta(q, \lambda) = q$ and $\delta(q, s w) = \delta(\delta(q, s), w)$ for $q \in Q$, $s \in \Sigma$ and $w \in \Sigma^*$.
  The language is defined by $\mathcal{L}_q(\mathcal{A}) = \{ w \in \Sigma^* \;|\; \delta(q, w) \in F \}$.
\item
  If $\mathcal{A}$ is a nondeterministic finite automaton, or NFA, then $\delta : Q \times \Sigma \to 2^Q$ gives the set of possible next states.
  We extend it to $Q \times \Sigma^*$ by $\delta(q, \lambda) = \{q\}$ and $\delta(q, s w) = \bigcup_{p \in \delta(q, s)} \delta(p, w)$ for $q \in Q$, $s \in \Sigma$ and $w \in \Sigma^*$.
  The language is defined by $\mathcal{L}_q(\mathcal{A}) = \{ w \in \Sigma^* \;|\; \delta(q, w) \cap F \neq \emptyset \}$.
\end{itemize}
An NFA can be converted into an equivalent DFA by the standard subset constructions.
A standard reference for DFAs and NFAs is~\cite{HopcroftUllman}.

Two states $p, q \in Q$ of $\mathcal{A}$ are equivalent, denoted $p \sim q$, if $\mathcal{L}_p(\mathcal{A}) = \mathcal{L}_q(\mathcal{A})$.
Every regular language $L \subseteq \Sigma^*$ is accepted by a unique DFA with the minimal number of states, which are all nonequivalent, and every other DFA that accepts $L$ has an equivalent pair of states.
Two words $v, w \in \Sigma^*$ are congruent by $L$, denoted $v \equiv_L w$, if for all $u_1, u_2 \in \Sigma^*$ we have $u_1 v u_2 \in L$ iff $u_1 w u_2 \in L$.
They are right-equivalent, denoted $v \sim_L w$, if for all $u \in \Sigma^*$ we have $v u \in L$ iff $w u \in L$.
The set of equivalence classes $\Sigma^* / {\equiv_L}$ is the syntactic monoid of $L$, and if $L$ is regular, then it is finite.
In that case the equivalence classes of $\sim_L$ can be taken as the states of the minimal DFA of $L$.

Let $\mathcal{P} : 2^{\Sigma^*} \to 2^{\Sigma^*}$ be a (possibly partially defined) operation on languages.
The (regular) state complexity of $\mathcal{P}$ is $f : \N \to \N$, where $f(n)$ is the maximal number of states in a minimal automaton of $\mathcal{P}(\mathcal{L}(\mathcal{A}))$ for an $n$-state DFA $\mathcal{A}$.

We say that a function $f : \N \to \R$ grows \emph{doubly exponentially} if there exist $a, b, c, d > 1$ with $a^{b^n} \leq f(n) \leq c^{d^n}$ for large enough $n$, and \emph{superexponentially} if for all $a > 1$, $f(n) > a^n$ holds for large enough $n$.

\section{Winning Sets}

In this section we define winning sets of binary languages, present the construction of the winning set of a regular language, and prove some general lemmas.
We defined the winning set informally at the beginning of Section~\ref{sec:Intro}.
Now we give a more formal definition which does not explicitly mention games.

\begin{definition}[Winning Set]
  Let $n \in \N$ and $T \subseteq \{0,1\}^n$ be arbitrary.
  The \emph{winning set} of $T$, denoted $W(T) \subseteq \{A, B\}^n$, is defined inductively as follows.
  If $n = 0$, then $T$ is either the empty set or $\{\lambda\}$, and $W(T) = T$.
  If $n \geq 1$, then $W(T) = \{ A w \;|\; w \in W(0^{-1} T) \cup W(1^{-1} T) \} \cup \{ B w \;|\; w \in W(0^{-1} T) \cap W(1^{-1} T) \}$.
  
  For a language $L \subseteq \{0, 1\}^*$, we define $W(L) = \bigcup_{n \in \N} W(L \cap \{0,1\}^n)$.
\end{definition}

For Alice to win on a turn order of the form $A w$, she has to choose either $0$ or $1$ as the first letter $v_0$ of the constructed word $v$, and then follow a winning strategy on the target set $v_0^{-1} T$ and turn order $w$.
On a word $B w$, Alice must have a winning strategy on $v_0^{-1} T$ and $w$ no matter how Bob chooses $v_0$.

\begin{example}
  Consider the regular language $L = (0+1)^* 0 1 1 (0+1)^*$ of all binary words that contain an occurrence of $0 1 1$, and the turn order $w = A A B A A B$.
  We claim that Alice has a winning strategy on $w$, so that $w \in W(L)$.
  On her first two turns she can play $0 1$, after which it is Bob's turn.
  If Bob plays a $1$, the constructed word $v \in \{0,1\}^6$ begins with $0 1 1$.
  If Bob plays a $0$ instead, Alice can follow with $1 1$, so $v$ begins with $0 1 0 1 1$.
  In both cases $v \in L$ and Alice wins.
  Note how the final $B$ plays no role in this analysis, and indeed $AABAA(A+B)^* \subset W(L)$.
\end{example}

A language $L$ over a linearly ordered alphabet $\Sigma$ is \emph{downward closed} if $v \in L$, $w \in \Sigma^{|v|}$ and $w_i \leq v_i$ for each $i = 0, \ldots, |v|-1$ always implies $w \in L$.

\begin{proposition}[Propositions~3.8 and~5.4 in~\cite{salo2014playing}]
  \label{prop:knownProperties}
  The winning set $W(L)$ of any $L \subseteq \{0,1\}^*$ is downward closed (with the ordering $A < B$) and satisfies $|W(L) \cap \{A,B\}^n| = |L \cap \{0,1\}^n|$ for all $n$.
  If $L$ is regular, then so is $W(L)$.
\end{proposition}

From a DFA $\mathcal{A}$, we can easily construct a nondeterministic automaton for $W(\mathcal{A})$.

\begin{definition}[Winning Set Automaton]
  \label{winningSetAuto}
  Let $\mathcal{A} = (Q, \{0,1\}, q_0, \delta, F)$ be a binary DFA.
  Define a ``canonical'' NFA for the winning set $W(\mathcal{L}(\mathcal{A}))$ as $\mathcal{A}' = (2^Q, \{A,B\}, \{q_0\}, \delta', 2^F)$, where
  \begin{align*}
    \delta'(S, A) & = \{\{ \delta(q, f(q)) \;|\; q \in S \} \;|\; f : S \to \{0,1\} \} \\
    \delta'(S, B) & = \{\{ \delta(q, b) \;|\; q \in S, b \in \{0,1\} \}\}.
  \end{align*}
  We usually work on the determinization of this NFA, which we denote by $W(\mathcal{A}) = (2^{2^Q}, \{A, B\}, \{\{q_0\}\}, \delta_W, F_W)$.
  Here $F_W = \{ \gst{G} \in 2^{2^Q} \;|\; \exists S \in \gst{G} : S \subseteq F \}$ and $\delta_W (\gst{G}, c) = \bigcup_{S \in \gst{G}} \delta' (S,c)$
  for $\gst{G} \subset 2^Q$ and $c \in \{A, B\}$.
\end{definition}

\begin{lemma}
  In the situation of Definition~\ref{winningSetAuto}, we have $\mathcal{L}(\mathcal{A}') = W(\mathcal{L}(\mathcal{A}))$.
\end{lemma}

\begin{proof}
  Let $S \in 2^Q$ and $w \in \{A,B\}^*$ be arbitrary.
  We prove by induction on $|w|$ that $w \in \mathcal{L}_S(\mathcal{A}')$ if and only if $w \in \bigcap_{q \in S} W(\mathcal{L}_q(\mathcal{A}))$.
  The result follows by considering $S = \{q_0\}$.
  The base case $w = \lambda$ is clear, since the final states of $\mathcal{A}'$ are exactly the subsets of $F$.
  
  Consider $w = A v$.
  If $A v \in \mathcal{L}_S(\mathcal{A}')$ then there is a function $f : S \to \{0,1\}$ with $v \in \mathcal{L}_T(\mathcal{A}')$ for $T = \{ \delta(q, f(q)) \;|\; q \in S \}$.
  By the induction hypothesis, for each $q \in S$ we have $v \in W(\mathcal{L}_{\delta(q, f(q))}(\mathcal{A})) = W(f(q)^{-1} \mathcal{L}_q(\mathcal{A}))$, so that $A v \in W(\mathcal{L}_q(\mathcal{A}))$.
  Conversely, if $A v \in W(\mathcal{L}_q(\mathcal{A}))$, then $v \in W(0^{-1} \mathcal{L}_q(\mathcal{A})) \cup W(1^{-1} \mathcal{L}_q(\mathcal{A})) = W(\mathcal{L}_{\delta(q, 0)}(\mathcal{A})) \cup W(\mathcal{L}_{\delta(q, 1)}(\mathcal{A}))$, so we may choose $f(q) \in \{0,1\}$ with $v \in W(\mathcal{L}_{\delta(q, f(q))}(\mathcal{A}))$.
  By the induction hypothesis, $v \in \mathcal{L}_T(\mathcal{A}')$, and then $A v \in \mathcal{L}_S(\mathcal{A}')$.
  
  Consider then $w \in B v$.
  We have $B v \in \mathcal{L}_S(\mathcal{A}')$ if and only if $v \in \mathcal{L}_T(\mathcal{A}')$ with $T = \{ \delta(q, b) \;|\; q \in S, b \in \{0,1\} \}$.
  By the induction hypothesis this is equivalent to $v \in \bigcap_{b \in \{0,1\}} W(\mathcal{L}_{\delta(q, b)}(\mathcal{A})) = \bigcap_{b \in \{0,1\}} W(b^{-1} \mathcal{L}_q(\mathcal{A}))$, hence to $B v \in W(\mathcal{L}_q(\mathcal{A}))$, for each $q \in S$.
\end{proof}

Intuitively, as Alice and Bob construct a word, they also play a game on the states of $\mathcal{A}$ by choosing transitions.
A state $\gst{G}$ of $W(\mathcal{A})$ is called a \emph{game state}, and it represents a situation where Alice can force the game to be in one of the sets $S \in \gst{G}$, and Bob can choose the actual state $q \in S$.
From the definition of $\mathcal{A}'$ it follows that exchanging the labels $0$ and $1$ on the two outgoing transitions of any one state of $\mathcal{A}$ does not affect $\delta'$.
In other words, the winning set of a DFA's language is independent of the labels of its transitions.

\begin{example}
  Consider $L=0^*1(0^*10^*1)0^*$, the language of words with an odd number of $1$-symbols. Its winning set is $W(L)=(A+B)^*A$, as the last player has full control of the parity of occurrences of $1$s. Figure \ref{exampleWinningSetAuto} shows
  the minimal DFA for $L$ and the NFA derived from it that recognizes $W(L)$. One can check that the language recognized by this NFA is indeed $(A+B)^*A$.
  Note how reading $A$ lets each state of a set evolve independently by $0$ or $1$, while $B$ makes both choices for all states simultaneously and results in one large set.

  \begin{figure}[htp]
    \begin{center}
      \begin{tikzpicture}
	[nonaccept/.style={circle,draw,inner sep=0cm,minimum size=0.75cm},
	 accept/.style={circle,draw,double,inner sep=0cm,minimum size=0.75cm},
	 every edge/.style={draw,->,>=stealth},
	 scale=0.45]
		\node [style=nonaccept] (a) at (0,0) {$a$};
		\node [style=accept] (b) at (4,0) {$b$};
		
		\node [style=nonaccept] (a2) at (10,0) {$\{a\}$};
		\node [style=nonaccept] (ab2) at (14,0) {$\{a,b\}$};
		\node [style=accept] (b2) at (18,0) {$\{b\}$};

		\draw [<->] (a) to node [midway,above] {$1$} (b);
		\draw [->, in=135, out=45, loop, looseness=5] (a) to node [midway,above] {$0$} (a);
		\draw [->, in=135, out=45, loop, looseness=5] (b) to node [midway,above] {$0$} (b);
		\draw [<-] (a) to (-2,0);

		\draw [<->, bend right] (a2) to node [midway,below] {$A$} (b2);
		\draw [->, in=135, out=45, loop, looseness=5] (a2) to node [midway,above] {$A$} (a2);
		\draw [->, in=135, out=45, loop, looseness=5] (b2) to node [midway,above] {$A$} (b2);
		\draw [->, in=135, out=45, loop, looseness=5] (ab2) to node [midway,above] {$A,B$} (ab2);
		\draw [<-] (a2) to (8,0);

                \draw [->,bend left] (a2) to node [midway,above] {$B$} (ab2);
		\draw [->,bend right] (b2) to node [midway,above] {$B$} (ab2);
		\draw [->,bend left=15] (ab2) to node [midway,above] {$A$} (a2);
		\draw [->,bend right=15] (ab2) to node [midway,above] {$A$} (b2);
\end{tikzpicture}
    \end{center}
    \caption{\label{exampleWinningSetAuto} A DFA for $L=0^*1(0^*10^*1)0^*$, and the derived NFA for $W(L)$.}
  \end{figure}
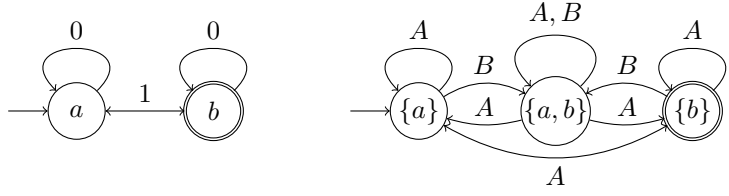
\end{example}

The following properties follow easily from the definition of $W(\mathcal{A})$.

\begin{lemma} \label{basicProperties}
  Let $\mathcal{A}$ be a binary DFA, $W(\mathcal{A})$ the winning set DFA from Definition~\ref{winningSetAuto},
  and $\delta_W$ the iterated transition function for $W(\mathcal{A})$.
  Let $\gst{G}$ and $\gst{H}$ be game states of $W(\mathcal{A})$, $R,S,T,V \subseteq 2^Q$ sets of states, and $w$ a word over $\{A,B\}$.
  \begin{roster}
  \item \label{unionEquiv}
    Sets in game states evolve independently:
    $\delta_W (\gst{G} \cup \gst{H}, w) = \delta_W (\gst{G}, w) \cup \delta_W (\gst{H}, w)$.
  \item \label{unionInsideEquiv}
    States in sets evolve almost independently:
    If $S, R \subseteq Q$ are disjoint, then $\delta_W(\{S \cup R\}, w) = \{ T \cup V \;|\; T \in \delta_W(\{S\}, w), V \in \delta_W(\{R\}, w) \}$.
  \item \label{subsetEquiv}
    Supersets can be removed from game states: 
    If $S,R \in \gst{G}$ and $S \subsetneq R$, then $\gst{G} \sim \gst{G} \setminus \{R\}$.
  \item \label{removeStateEquiv}
    Sets containing nonaccepting sink states can be removed from game states:
    If $S \in \gst{G}$ and some $q \in S$ has no path to a final state, then $\gst{G} \sim \gst{G} \setminus \{S\}$.
  \item \label{finalSinkEquiv}
    Accepting sink states can be removed from sets:
    If $S \in \gst{G}$ and there is a sink state $q \in S \cap F$, then $\gst{G} \sim ( \gst{G} \setminus \{S\}) \cup \{S \setminus \{q\} \}$.
  \end{roster}
\end{lemma}

\begin{proof}
  \begin{roster}
  \item
    This follows from the definition of $\delta_W$ in the powerset construction.
  \item
    Note first that $\delta_W(\{S\}, w) = \delta'(S, w)$ in the notation of Definition~\ref{winningSetAuto}.
    By item~\ref{unionEquiv} it is enough to consider the case $|w| = 1$.
    For $T \subset Q$ and $f : T \to \{0,1\}$, denote $\delta(T, f) = \{ \delta(q, f(q)) \;|\; q \in T \}$.
    For $w = A$ we have
    \begin{align*}
      \delta'(S \cup R, A) = {} & \{ \delta(S \cup R, f) \;|\; f : S \cup R \to \{0,1\} \} \\
      {} = {} & \{ \delta(S, f) \cup \delta(R, g) \;|\; f : S \to \{0,1\}, g : R \to \{0,1\} \} \\
      {} = {} & \{ T \cup V \;|\; T \in \delta'(S, A), V \in \delta'(R, A) \}.
    \end{align*}
    For $w = B$ we have $\delta'(S \cup R, B) = \{ \delta(q, b) \;|\; q \in S \cup R, b \in \{0,1\} \} = \{ \delta(q, b) \;|\; q \in S, b \in \{0,1\} \} \cup \{ \delta(q, b) \;|\; q \in R, b \in \{0,1\} \} = \delta'(S, B) \cup \delta'(R, B)$.
  \item
    Denote $\gst{H} = \gst{G} \setminus \{R\}$ and $T = R \setminus S$, so that $\gst{G} = \gst{H} \cup \{S \cup T\}$ and the unions are disjoint.
    By items~\ref{unionEquiv} and~\ref{unionInsideEquiv} we have
    \[
      \delta_W(\gst{G}, w) = \delta_W(\gst{H}, w) \cup \{ U \cup V \;|\; U \in \delta_W(\{S\}, w), V \in \delta_W(\{T\}, w)\}.
    \]
    If $\delta_W(\gst{H}, w)$ is accepting, then some $X \in \delta_W(\gst{H}, w)$ satisfies $X \subseteq F$, and $\delta_W(\gst{G}, w)$ is also accepting.
    Conversely, if $\delta_W(\gst{G}, w)$ is accepting, then some $X \in \delta_W(\gst{G}, w)$ satisfies $X \subseteq F$.
    If $X \in \delta_W(\gst{H}, w)$, then we are done; otherwise $X = U \cup V$ with $U \in \delta_W(\{S\}, w)$ and $V \in \delta_W(\{T\}, w)$.
    Then $U \subseteq F$, and since $S \in \gst{H}$, the game state $\delta_W(\gst{H}, w)$ is accepting.
  \item
    The game state $\delta_W(\{S\}, w)$ is not accepting for any $w \in \{A, B\}^*$, since each of its elements contains $\delta(q, v) \notin F$ for some $v \in \{0,1\}^{|w|}$.
    Thus $\delta_W(\gst{G}, w) = \delta_W(\gst{G} \setminus \{S\}, w) \cup \delta_W(\{S\}, w)$ is accepting if and only if $\delta_W(\gst{G} \setminus \{S\}, w)$ is.
  \item
    Denote $R = S \setminus \{ q \}$ and let $w \in \{A,B\}^*$ be arbitrary.
    Since $q$ is a sink state, $\delta_W(\{\{q\}\}, w) = \{\{q\}\}$.
    By item~\ref{unionInsideEquiv}, $\delta_W(\{S\}, w) = \{ T \cup V \;|\; T \in \delta_W(\{R\}, w), V \in \delta_W(\{\{q\}\}, w) \} = \{ T \cup \{q\} \;|\; T \in \delta_W(\{R\}, w) \}$.
    Since $q \in F$, this game state is accepting if and only if $\delta_W(\{R\}, w)$ is.
    Thus $\{S\} \sim \{R\}$, and item~\ref{unionEquiv} implies $\gst{G} \sim (\gst{G} \setminus \{S\}) \cup \{R\}$.
  \end{roster}
\end{proof}

Define a partial order on $2^{2^Q}$ as follows: $\gst{G} \leq \gst{H}$ if for each $S \in \gst{G}$ there exists $R \in \gst{H}$ with $R \subseteq S$.
By Lemma~\ref{basicProperties}\ref{unionEquiv} and~\ref{unionInsideEquiv}, the transition function $\delta_W$ is monotonic with respect to ${\leq}$ in the sense that $\gst{G} \leq \gst{H}$ implies $\delta_W(\gst{G}, w) \leq \delta_W(\gst{H}, w)$ for all $w \in \{A,B\}^*$.
Furthermore, if $\gst{G} \leq \gst{H}$ and $\gst{G}$ is accepting, then some $S \in \gst{G}$ and $R \in \gst{H}$ satisfy $R \subseteq S \subseteq F$, hence $\gst{H}$ is also accepting.

The next lemma helps prove equivalences of game states and words.

\begin{lemma} \label{lemmasEquivIncl}
  Recall the assumptions of Lemma~\ref{basicProperties}.
  \begin{roster}
  \item \label{equivChainDoubleIncl}
    If $\gst{G} \leq \gst{H}$ and $\gst{H} \leq \gst{G}$, then $\gst{G} \sim \gst{H}$.
  \item  \label{equivChainSingletons}
    Let $v, w \in \{A, B\}^*$.
    If for all $q \in Q$, the game states $\delta_W (\{\{q\}\},v)$ and $\delta_W (\{\{q\}\},w)$ are either both accepting or both rejecting, then $v \equiv_{W(\mathcal{L}(\mathcal{A}))} w$.
  \end{roster}
  
\end{lemma}

\begin{proof}
  \begin{roster}
  \item
    This is immediate from the monotonicity of $\delta_W$.
  \item
    Let $\gst{G} \in 2^{2^Q}$ be a game state and suppose $\delta_W(\gst{G}, v)$ is accepting, so there exists $P \in \delta_W(\gst{G}, v)$ consisting of accepting states of $\mathcal{A}$.
    We claim $\delta_W(\gst{G}, w)$ is also accepting.
    We have $\delta_W(\gst{G}, v) = \bigcup_{S \in \gst{G}} \delta_W(\{S\}, v)$ by definition, and similarly for $w$, so we may assume $\gst{G} = \{S\}$ is a singleton.
    By Lemma~\ref{basicProperties}\ref{unionInsideEquiv}, for each $q \in S$ there exists $R_q \in \delta_W(\{\{q\}\}, v)$ such that $P = \bigcup_{q \in S} R_q$.
    In particular $R_q \subseteq F$, so each $\delta_W(\{\{q\}\}, v)$ is accepting.
    By assumption $\delta_W(\{\{q\}\}, w)$ is also accepting, so there exists $T_q \in \delta_W(\{\{q\}\}, w)$ with $T_q \subseteq F$.
    By Lemma~\ref{basicProperties}\ref{unionInsideEquiv} we have $\bigcup_{q \in S} T_q \in \delta_W(\{S\}, w)$, so $\delta_W(\{S\}, w)$ is accepting.
    This shows $v \equiv_{W(\mathcal{L}(\mathcal{A}))} w$.
  \end{roster}
\end{proof}

Recall the \emph{Dedekind numbers} $D(n)$, which count the number of antichains of subsets of $\{1, \ldots, n\}$ with respect to set inclusion.
Their growth is doubly exponential: $a^{a^n} < D(n) < b^{b^n}$ holds for large enough $n$ if $a < 2 < b$.
This follows from $\binom{n}{\lceil n/2 \rceil} \leq \log_2 D(n) \leq (1 + O(\log n / n)) \binom{n}{\lceil n/2 \rceil}$ \cite{KleitmanMarkowskyDedekind75} and the well known asymptotic formula $\binom{n}{\lceil n/2 \rceil} = \Theta(2^n / \sqrt{n})$.


\begin{proposition} \label{prop:upperBound}
  Let $\mathcal{A}$ an $n$-state DFA. 
  The number of states in the minimal DFA for $W(\mathcal{L}(\mathcal{A}))$ is at most the Dedekind number $D(n)$.
\end{proposition}
\begin{proof}
  Every game state is equivalent to an antichain by Lemma \ref{basicProperties}\ref{subsetEquiv}, so the number of nonequivalent game states is at most $D(n)$.
\end{proof}

We have computed the exact state complexity of the winning set operation for DFAs with at most $5$ states; the $6$-state case is no longer feasible with our program and computational resources.
The sequence begins with $1, 4, 16, 62, 517$.

\section{Doubly Exponential Lower Bound}

In this section we construct a family of automata for which the number of states in the minimal
winning set automaton is doubly exponential.
The idea is to reach any desired antichain of subsets of a special subset of states, and then 
to make sure these game states are nonequivalent.
To do this we split the automaton into several components.
First we present a ``subset factory gadget''
that allows to reach any 
set of the form $\{S\}$ where $S$ is a subset of a specific length-$n$ path in the automaton. This gadget will be used several times to accumulate subsets in the game state.
Then we present a ``testing gadget'' that lets us distinguish between game states by whether they contain a (subset of a) given set or not.

Recall that the transition labels of a binary DFA
are irrelevant to the winning set of its language.
In this section we define automata by describing their graphs, and a node with two outgoing transitions can have them arbitrary labeled by $0$ and $1$.
Incoming and outgoing transitions in the figures indicate how the gadgets connect to the rest of the automaton.

\begin{lemma}[Subset factory gadget]
  \label{lem:subsetFactoryGadget}
  Let $\GenSubset_n$ be the graph in Figure~\ref{subsetFactoryGadget}.
  For $i \in \{1, \ldots, n\}$, denote $o_i = e_{2n + i -1}$ (the $n$ rightmost states labeled by $e$).
  For $S \subseteq \{1,\ldots,n\}$, let $w^\mathrm{gen}_S$ be the concatenation $w_1 w_2 \dots w_n$ where $w_i = BA$ if $i \in S$, and $w_i = AB$ if $i \notin S$.
  Then $\delta_W ( \{\{b_1\}\}, w_S^\mathrm{gen}) )  \sim  \{ \{ o_i \;|\; i \in S \} \} $ for each binary DFA that contains $\GenSubset_n$ as a subgraph.
\end{lemma}

The idea is that at step $i$, reading $B$ adds $c_i$ to each subset of the game state, 
and then reading $A$ avoids the sink $s_i$.
On the other hand reading $A$ creates two versions of each subset of the game state, one that continues on the upper row, and one that falls into
the sink $s_i$ when $B$ is read, and can then be ignored.

\begin{figure}[htp]
  \begin{center}
    \begin{tikzpicture}
      [nonaccept/.style={circle,draw,inner sep=0cm,minimum size=0.75cm},
      accept/.style={circle,draw,double,inner sep=0cm,minimum size=0.75cm},
	 every edge/.style={draw,->,>=stealth},
	 scale=0.4]
		\node [style=nonaccept] (0) at (-4, 2) {$b_1$};
		\node [style=nonaccept] (1) at (-1, 2) {$d_1$};
		\node [style=nonaccept] (2) at (5.75, 2) {$b_{n-1}$};
		\node [style=nonaccept] (3) at (9, 2) {$d_{n-1}$};
		\node [style=nonaccept] (4) at (-4, -1) {$c_1$};
		\node [style=nonaccept] (5) at (-6.25, -1) {$s_1$};
		\node [style=nonaccept] (6) at (3.25, -1) {$s_{n-1}$};
		\node [style=nonaccept] (7) at (5.75, -1) {$c_{n-1}$};
		\node [style=nonaccept] (8) at (11.25, -1) {$s_n$};
		\node [style=nonaccept] (9) at (13.5, 2) {$b_n$};
		\node [style=nonaccept] (10) at (13.5, -1) {$c_n$};
		\node [style=nonaccept] (11) at (13.5, -4) {$e_{3n - 2}$};
		\node [style=nonaccept] (12) at (5.75, -4) {$e_{3n-5}$};
		\node [style=nonaccept] (13) at (8.5, -4) {$e_{3n-4}$};
		\node [style=nonaccept] (14) at (11, -4) {$e_{3n-3}$};
		\node [style=nonaccept] (15) at (0.75, -4) {$e_3$};
		\node [style=nonaccept] (16) at (-1.5, -4) {$e_2$};
		\node [style=nonaccept] (17) at (-4, -4) {$e_1$};
		\node (18) at (3, 2) {$\phantom{I}$};
		\node (19) at (0.75, 2) {$\phantom{I}$};
		\node (20) at (2.5, -4) {$\phantom{I}$};
		\node (21) at (4, -4) {$\phantom{I}$};
		\node (22) at (2, 2) {$\cdots$};
		\node (23) at (3.25, -4) {$\cdots$};
		\node (24) at (-6, 2) {};
		\node [style=accept,inner sep=0cm,minimum size=0.75cm] (25) at (17.5, 2) {$b_{n+1}$};
		\node (26) at (17.5, -4) {};
		\node (27) at (17.5, -4) {$\phantom{I}$};
		
		\draw (0) edge (1);
		\draw (2) edge (3);
		\draw (0) edge (4);
		\draw (4) edge (5);
		\draw [in=130, out=-130, loop] (5) edge (5);
		\draw [in=145, out=-145, loop, looseness=5] (5) edge (5);
		\draw (7) edge (6);
		\draw [in=130, out=-130, loop] (6) edge (6);
		\draw [in=145, out=-145, loop, looseness=5] (6) edge (6);
		\draw (2) edge (7);
		\draw [out=10,in=170] (3) edge (9);
		\draw [out=-10,in=-170] (3) edge (9);
		\draw (9) edge (10);
		\draw (10) edge (8);
		\draw [in=130, out=-130, loop] (8) edge (8);
		\draw [in=145, out=-145, loop, looseness=5] (8) edge (8);
		\draw (7) edge (12);
		\draw [in=170, out=10] (12) edge (13);
		\draw [in=-170, out=-10] (12) edge (13);
		\draw (10) edge (11);
		\draw [in=170, out=10] (13) edge (14);
		\draw [in=-170, out=-10] (13)edge (14);
		\draw [in=170, out=10] (14) edge (11);
		\draw [in=-170, out=-10] (14) edge (11);
		\draw (4) edge (17);
		\draw [in=170, out=10] (17) edge (16);
		\draw [in=-170, out=-10] (17) edge (16);
		\draw [in=170, out=10] (16) edge (15);
		\draw [in=-170, out=-10] (16) edge (15);
		\draw [in=170, out=10] (1) edge (19);
		\draw [in=-170, out=-10] (1) edge (19);
		\draw [in=170, out=10] (18) edge (2);
		\draw [in=-170, out=-10] (18) edge (2);
		\draw [in=170, out=10] (15) edge (20);
		\draw [in=-170, out=-10] (15) edge (20);
		\draw [in=170, out=10] (21) edge (12);
		\draw [in=-170, out=-10] (21) edge (12);
		\draw (24) edge (0);
		\draw (9) edge (25);
		\draw [in=-135, out=-45, loop] (25) edge (25);
		\draw [in=-120, out=-60, loop, looseness=5] (25) edge (25);
		\draw [in=170, out=10] (11) edge (27);
		\draw [in=-170, out=-10] (11) edge (27);
\end{tikzpicture}

  \end{center}
  \caption{\label{subsetFactoryGadget} $\GenSubset_n$, the subset factory gadget.}
\end{figure}
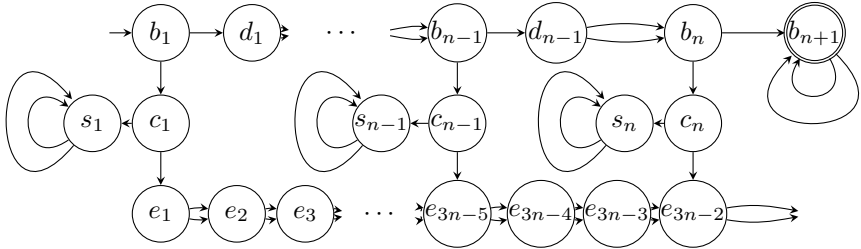

\begin{proof}
  Denote $f_i = e_{3 i - 2} $.
  For $i \in \{1, \ldots, n\}$ and $S \subseteq \{1, \ldots, i-1\}$, denote $S_i = \{e_{2 i - 4 + j} \;|\; j \in S\}$.
  Consider the game state $\gst{H}(i,S) = \{\{b_i\} \cup S_i\}$.
  If the automaton reads $A B$, the resulting game state is
  \[
    \delta_W(\gst{H}(i,S), A B) = \{\{s_i, e_{3 i}\} \cup S_{i+1}, \{b_{i+1}\} \cup S_{i+1}\} \sim \gst{H}(i+1, S) .
  \]
  In the case of $B A$ we instead have
  \[
    \delta_W(\gst{H}(i,S), B A) = \{ \{b_{i+1}, s_i\} \cup S_{i+1}, \{b_{i+1}, f_i\} \cup S_{i+1} \} \sim \gst{H}(i+1, S \cup \{i\}) .
  \]
  In both cases the final steps follow from Lemma~\ref{basicProperties}\ref{removeStateEquiv}.
  By Lemma~\ref{basicProperties}\ref{finalSinkEquiv} we also have $\gst{H}(n+1, S) \sim \{S_{n+1}\}$ since $b_{n+1}$ is an accepting sink state.
  
  For $i = n$ we now have $\delta_W(\{\{b_1\}\}, w^\mathrm{gen}_S) = \{S_{n+1}\} = \{\{o_i \;|\; i \in S\}\}$.
\end{proof}

\begin{lemma}[Game state factory gadget]\label{subsetFactory}
  Let $\GenState_n$ be the graph in Figure~\ref{stateFactoryGadget} and $\mathcal{A}$ any DFA over $\{0,1\}$ that contains it.
  For all $\gst{G} = \{ S_1, \ldots, S_\ell \}$ where each $S_i \subseteq \{r_1,\ldots, r_n\}$, let $w^\mathrm{gen}_{\gst{G}} \in \{A,B\}^{\ell (3 n + 1)}$ be the concatenation of $A w^\mathrm{gen}_{S_i} A^n$ for $i \in \{0, \ldots, \ell\}$.
  Then
  $\delta_W (\{\{a_1\}\},w^\mathrm{gen}_{\gst{G}})  \sim \gst{G} \cup \{\{a_1\}\} \cup \gst{G}'$
  for some game state $\gst{G}'$ that does not contain a subset of the states of $\GenState_n$.

  \begin{figure}[htp]
    \begin{center}
      \begin{tikzpicture}
	[nonaccept/.style={circle,draw,inner sep=0cm,minimum size=0.75cm},
	 accept/.style={circle,draw,double,inner sep=0cm,minimum size=0.75cm},
	 every edge/.style={draw,->,>=stealth},
	 scale=0.4]

	\begin{scope}[shift={(-8cm,-7cm)},rotate=90]
	\node [nonaccept] (circ1) at (270:4cm) {$a_1$};
	\node [nonaccept] (circ2) at (310:4cm) {$a_2$};
	\node [nonaccept] (circ3) at (350:4cm) {$a_3$};
	\node [nonaccept] (circ4) at (30:4cm) {$a_4$};
	\node [nonaccept] (circ5) at (70:4cm) {$a_5$};
	\node [nonaccept] (circ6) at (110:4cm) {$a_6$};
	\node (circ7) at (150:4cm) {$\cdots$};
	\node [nonaccept] (circ8) at (190:4cm) {$a_{3n}$};
	\node [nonaccept] (circ9) at (230:4cm) {$a_{3n+1}$};
	\end{scope}

	\draw (circ1) edge (circ2);
	\foreach \i/\j/\ang in {2/3/40, 3/4/80, 4/5/120, 5/6/160, 6/7/200, 7/8/240, 8/9/280, 9/1/320}{
		 \draw (circ\i) edge (circ\j);
		 \draw [in=310+\ang, out=90+\ang] (circ\i) edge (circ\j);
	}

	\node [draw, rectangle, minimum height=1cm, minimum width=2.5cm] (gensubset) at (1,-7) {$\GenSubset_n$};
	\draw (circ1) edge (gensubset);

	\begin{scope}[shift={(11cm,-7cm)}]
	\node [nonaccept] (circ21) at (180:4cm) {$r_1$};
	\node [nonaccept] (circ22) at (220:4cm) {$r_2$};
	\node [nonaccept] (circ23) at (260:4cm) {$r_3$};
	\node [nonaccept] (circ24) at (300:4cm) {$r_4$};
	\node [nonaccept] (circ25) at (340:4cm) {$r_5$};
	\node [nonaccept] (circ26) at (20:4cm) {$r_6$};
	\node (circ27) at (60:4cm) {$\cdots$};
	\node [nonaccept] (circ28) at (100:4cm) {$r_{3n}$};
	\node [nonaccept] (circ29) at (140:4cm) {$r_{3n+1}$};
	\end{scope}

	\draw [in=170,out=10] (gensubset) edge (circ21);
	\draw [in=-170,out=-10] (gensubset) edge (circ21);
	
	\draw (circ21) edge (circ22);
	\foreach \i/\j/\ang in {2/3/40, 3/4/80, 4/5/120, 5/6/160, 6/7/200, 7/8/240, 8/9/280, 9/1/320}{
		 \draw (circ2\i) edge (circ2\j);
		 \draw [in=130+\ang, out=270+\ang] (circ2\i) edge (circ2\j);
	}

	\node [right=1cm] (T) at (circ21) {};
	\draw (circ21) edge (T);

\end{tikzpicture}
    \end{center}
    
    \caption{\label{stateFactoryGadget} $\GenState_n$, the game state factory gadget.}
  \end{figure}
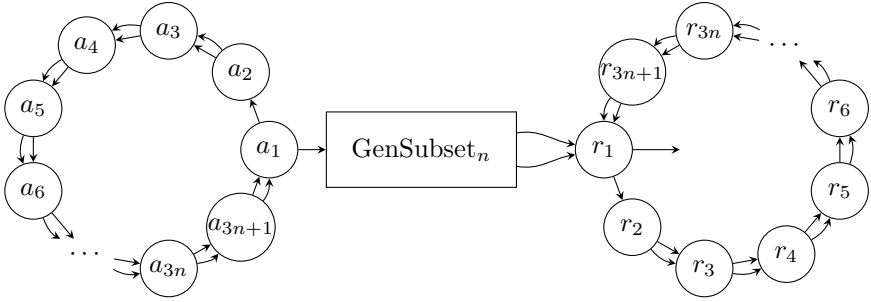

\end{lemma}

The idea is to successively add new sets $S_i$ to the game state, while previously made subsets will wait by rotating in the $r$-cycle.
A singleton set rotates in the $a$-cycle so that reading $A$ from the state $a_1$ creates a new singleton set in the subset factory gadget.
The word $w^\mathrm{gen}_{S_i}$ transforms it into a set of the correct form, and then reading $A^n$ both moves this new subset to the $r$-cycle with the previously created sets and rotates the singleton set back to $a_i$.

\begin{proof}
  Suppose we have reached a game state of the form $\gst{H} = \{ \{a_1\},\{r_i \;|\; i \in S_1\} , 
  \ldots,  \{r_i \;|\; i \in S_k\}\} \cup \gst{G}'$ where $\gst{G}'$ does not contain any subset of $\GenState_n$.
  We prove that by reading $A w^\mathrm{gen}_{S_{k+1}} A^n$,
  we reach a game state of the form $\{ \{a_1\},\{r_i \;|\; i \in S_1\} , 
  \ldots,  \{r_i \;|\; i \in S_{k+1}\}\} \cup \gst{G}''$.
  We analyze the elements of $\gst{H}$ separately.
  \begin{itemize}
  \item Because $|A w^{gen}_{S_{k+1}}A^{n}| = 3n+1$
    is the size of the rightmost cycle, we have $\delta_W (\{ \{r_i \;|\; i \in S_j\} \}, A w^{gen}_{S_{k+1}} A^{n}) \sim \{ \{r_i : i \in S_j\}\} \cup \gst{G}'_j$ for each $j \leq k$, where each set in $\gst{G}'_j$ contains a state outside the gadget.
    
  \item The game state $\{\{a_1\}\}$ first evolves into $\delta_W(\{\{a_1\}\}, A) = \{ \{a_2\},\{b_1\}\}$.
    The component $\{\{a_2\}\}$ becomes $\{\{a_1\}\}$ when we read $w^{gen}_{S_{k+1}}A^{n}$.
    As for $\{\{b_1\}\}$, Lemma~\ref{subsetFactory} gives $\delta_W(\{\{b_1\}\}, w^{gen}_{S_{k+1}} )  = 
    \{ \{ o_i \;|\; i\in S_{k+1}\}\}$,
    and then
    $\delta_W(\{ \{ o_i \;|\; i\in S_{k+1}\} \}, A^n ) = 
    \{ \{ r_{i} \;|\; i \in S_{k+1}\} \} \cup \gst{G}''$
    where every set in $\gst{G}''$ contains a state outside of $\GenState_n$.
  \item
    The game state $\gst{G}'$ evolves into some $\gst{G}'''$ each of whose sets contains a state not in $\GenState_n$, since the gadget cannot be re-entered.
  \end{itemize}
  
  By Lemma~\ref{basicProperties}\ref{unionEquiv} we have 
  $\delta_W (\gst{H}, A w^\mathrm{gen}_{S_{k+1}} A^{n})
  \sim \{ \{a_1\},\{r_{i} \;|\; i \in S_1\}, 
  \ldots,  \{r_{i} \;|\; i \in S_{k+1}\}\} \cup \gst{G}'' \cup \gst{G}''' \cup \bigcup_j \gst{G}'_j$.
  We obtain $w^\mathrm{gen}_{\gst{G}}$ as the concatenation of these words.
\end{proof}

\begin{lemma}[Testing gadget] \label{lem:testingGadget}
  Let $\Testing_n$ be the graph in Figure~\ref{measuringGadget}.
  \begin{roster}
  \item For $P \subseteq \{1, \dots, n\}$, define $w ^\mathrm{test}_P \in \{A,B\}^n$ by $w^\mathrm{test}_P[i] = A$ iff $n-i+1 \in P$. Then
    for each $I \subseteq \{1, \dots, n\}$, the game state $\delta_W (  \{\{ q_i \;|\; i \in I\}\}, w^\mathrm{test}_P)$ is accepting iff $I \subseteq P$.
  \item Let $V$ be the set of nodes of the graph $\Testing_n$.
    Then for all $\gst{G} \in 2^{2^{V}}$ and $w\in \{A, B\}^{\geq 2n}$, the game state $\delta_W (\gst{G},w)$ is not accepting.
  \end{roster}
  
  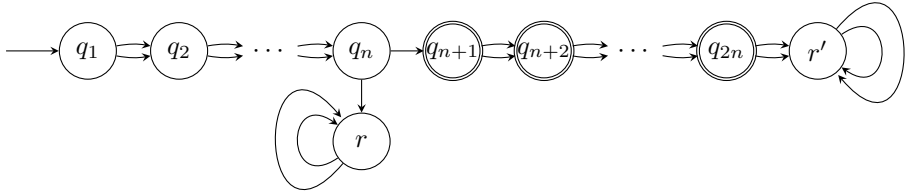
\begin{figure}[htp]
    \begin{center}
      \begin{tikzpicture}
	[nonaccept/.style={circle,draw,inner sep=0cm,minimum size=0.75cm},
	 accept/.style={circle,draw,double,inner sep=0cm,minimum size=0.75cm},
	 every edge/.style={draw,->,>=stealth},
	 scale=0.4]

	 \node (x) at (-3,0) {};
	 \node [nonaccept] (q1) at (0,0) {$q_1$};
 	 \node [nonaccept] (q2) at (3,0) {$q_2$};
 	 \node (q3) at (6,0) {$\cdots$};
 	 \node [nonaccept] (q4) at (9,0) {$q_n$};

	 \node [nonaccept] (r) at (9,-3) {$r$};

	 \node [accept] (a1) at (12,0) {$q_{n+1}$};
	 \node [accept] (a2) at (15,0) {$q_{n+2}$};
	 \node (a3) at (18,0) {$\cdots$};
	 \node [accept] (a4) at (21,0) {$q_{2 n}$};

	 \node [nonaccept] (r2) at (24,0) {$r'$};

	 \draw (x) edge (q1);

	 \foreach \i/\j in {q1/q2,q2/q3,q3/q4,a1/a2,a2/a3,a3/a4,a4/r2}{
	 	  \draw [in=170, out=10] (\i) edge (\j);
		  \draw [in=-170, out=-10] (\i) edge (\j);
	 }
	 \draw (q4) edge (r);
	 \draw (q4) edge (a1);

	\draw [in=130, out=-130, loop] (r) edge (r);
	\draw [in=145, out=-145, loop, looseness=5] (r) edge (r);

	\draw [in=-50, out=50, loop] (r2) edge (r2);
	\draw [in=-35, out=35, loop, looseness=5] (r2) edge (r2);

\end{tikzpicture}
    \end{center}

    \caption{\label{measuringGadget} $\Testing_n$, the testing gadget.}
  \end{figure}

\end{lemma}

The idea is that reading $A$ or $B$ moves the game state toward $r'$, except when the set contains the state $q_n$ and $B$ is read, causing it to fall into the sink $r$.

\begin{proof}
  \begin{roster}
  \item
    For $I \subseteq \{1, \ldots, 2 n\}$, denote $S_I = \{ q_i \;|\; i \in I \}$.
    If $2 n \notin I$, let $J = \{ i+1 \;|\; i \in I \}$.
    A simple case analysis together with Lemma~\ref{basicProperties}\ref{removeStateEquiv} and~\ref{finalSinkEquiv} shows that $\delta_W(\{S_I\}, A) \sim \{S_J\}$ and
    \[
      \delta_W(\{S_I\}, B) =
      \begin{cases}
        \emptyset & \mbox{if } n \in I, \\
        \{S_J\} & \mbox{otherwise.}
      \end{cases}
    \]
    Then $\gst{G} = \delta_W (\{S_I\}, w^\mathrm{test}_P)$ is accepting if and only if $\gst{G} \sim \{\{q_{i+n} \;|\; i \in I\}\}$.
    This is equivalent to $w[n-i+1] = A$ for all $i \in I$, i.e. $I \subseteq P$.
  \item If $\gst{G} \in 2^{2^{V}}$ and $w$ with $|w| \geq 2n$,
    then every $S \in \delta_W (\gst{G},w)$ satisfies $S \subseteq \{r,r'\}$.
  \end{roster}
\end{proof}

\begin{theorem}
  For each $n > 0$ there exists a DFA
  $\mathcal{A}_n$ over $\{0,1\}$ with $15n + 3$ states such that
  the minimal DFA for $W(\mathcal{L}(\mathcal{A}_n))$ has a least $D(n)$ states.
\end{theorem}

Together with Proposition~\ref{prop:upperBound}, this implies that the state complexity of $W$ restricted to regular languages grows doubly exponentially.

\begin{proof}
  Let $\mathcal{A}_n$ be the DFA obtained by combining $\Testing_n$ with the outgoing arrow of $\GenState_n$ and assigning $a_1$ as the initial state.
  For an antichain $\gst{G}$ on the powerset of $\{r_1, \ldots, r_n\}$, let $X_\gst{G} = \delta_W (\{\{a_1\}\}, w^\mathrm{gen}_\gst{G})$.
  Lemma~\ref{subsetFactory} gives $X_\gst{G} \sim \{\{a_1\}\} \cup \gst{G} \cup \gst{G}'$ where each set in $\gst{G}'$ contains a state of $\Testing_n$.

  We show that distinct antichains $\gst{G}$ result in nonequivalent states.
  Let $P \subseteq \{1, \ldots, n\}$ and consider $X'_\gst{G} = \delta_W(X_\gst{G}, A B^{2 n} A^{2 n+1} w^\mathrm{test}_P)$.
  We claim that $X'_\gst{G}$ is accepting iff some element of $\gst{G}$ is a subset of $\{ r_i \;|\; i \in P \}$.
  By Lemma~\ref{basicProperties}\ref{unionEquiv} we may analyze the components of $X_\gst{G}$ separately.
  \begin{itemize}
  \item
    We have $\delta_W(\{\{a_1\}\}, A) = \{\{a_2\},\{b_1\}\}$.
    The part $\{b_1\}$ is destroyed by the sink state $s_1$ when we read $B$s, and the part $\{a_2\}$ rotates in the $a$-cycle without encountering accepting states.
  \item
    Each set of $\gst{G}'$ contains a state of $\Testing_n$, which will reach one of the nonaccepting sinks $r$ or $r'$.
  \item
    The game state $\delta_W(\gst{G}, A B^{2 n} A^{2 n+1})$ consists of the sets $\{ q_i \;|\; r_i \in S \}$ for $S \in \gst{G}$, as well as sets that contain at least one element of $\{r_2, \ldots, r_{n+1}\}$.
    The latter will rotate in the $r$-cycle.
    By Lemma~\ref{lem:testingGadget}, the former sets produce an accepting game state in $X'_\gst{G}$ iff some $S \in \gst{G}$ is a subset of $\{ r_i \;|\; i \in P \}$.
  \end{itemize}
  We have found $D(n)$ nonequivalent states in $W(\mathcal{A})$.
\end{proof}


\section{Case of the Bounded Regular Languages}

In this section we prove an upper bound on the complexity of the winning set of a bounded regular language, which is lower then the one proved for unrestricted regular languages, but still superexponential.
We do not know if the true state complexity is superexponential for this class.
Our proof technique is based on tracing the evolution of individual states of a DFA $\mathcal{A}$ in the winning set automaton $W(\mathcal{A})$ when reading several $A$-symbols in a row.

Our motivation for studying bounded regular languages comes from the fact that they correspond to so-called \emph{zero entropy sofic shifts} in symbolic dynamics, which are defined by the number of words of given length that occur in them.
Recall that the definition of winning sets in~\cite{salo2014playing} was motivated by the study of entropy of shift spaces.
Since the winning set operation preserves entropy, in particular it preserves the property of having entropy zero.
See Section 4 of~\cite{PavlovSchraudner15} for discussion on zero entropy sofic shifts and their structure.
Analogously, since bounded regular languages are exactly those for which the number of words of given length grows polynomially, and the winning set operation preserves this number by Proposition~\ref{prop:knownProperties}, the winning set of a binary bounded language is also bounded.

\begin{definition}[Histories of Game States]
  \label{def:History}
  Let $\mathcal{A} = (Q, \{0,1\}, q_0, \delta, F)$ be a DFA.
  Let $\gst{G} \in 2^{2^Q}$ be a game state of $W(\mathcal{A})$, and for each $i \geq 0$, let $\gst{G}_i \sim \delta_W(\gst{G}, A^i)$ be the game state with all supersets removed as per Lemma~\ref{basicProperties}\ref{subsetEquiv}.
  A \emph{history function} for $\gst{G}$ is a function $h$ that associates to each $i > 0$ and each set $S \in \gst{G}_i$ a \emph{parent set} $h(i, S) \in \gst{G}_{i-1}$, and to each state $q \in S$ a set of \emph{parent states} $h(i, S, q) \subseteq h(i, S)$ such that
  \begin{itemize}
  \item $\{q\} \in \delta_W(\{h(i,S,q)\}, A)$ for all $q \in S$, and
  \item $h(i,S)$ is the disjoint union of $h(i, S, q)$ for $q \in S$.
  \end{itemize}
  Note that this implies $S \in \delta_W(\{h(i,S)\}, A)$ for each $i$.
  
  The \emph{history} of a set $S \in \gst{G}_i$ from $i$ under $h$ is the sequence $S_0, S_1, \ldots, S_i = S$ with $S_{j-1} = h(j, S_j)$ for all $0 < j \leq i$.
  A history of a state $q \in S$ in $S$ under $h$ is a sequence $q_0, \ldots, q_i = q$ with $q_{j-1} \in h(j, S_j, q_j)$ for all $0 < j \leq i$.
\end{definition}

Every game state has at least one history function: each $S \in \gst{G}_i$ has at least one set $R \in \gst{G}_{i-1}$ with $S \in \delta_W(\{R\}, A)$, so we can choose $R = h(i, S)$, and similarly for the $h(i, S, q)$.
It can have several different history functions, and each of them defines a history for each set $S$.
A state of $S$ can have several histories under a single history function.

For the rest of this section, we fix an $n$-state DFA $\mathcal{A} = (Q, \{0,1\}, q_0, \delta, F)$ that recognizes a bounded binary language and has disjoint cycles.
Let the lengths of the cycles be $k_1, \ldots, k_p$, and let $\ell$ be the number of states not part of any cycle.

We define a preorder ${\leq}$ on the state set $Q$ by reachability: $p \leq q$ holds if and only if there is a path from $p$ to $q$ in $\mathcal{A}$.
For two history functions $h, h'$ of a game state $\gst{G}$, we write $h \leq h'$ if for each $i > 0$, each $S \in \gst{G}_i$ and each $q \in S$, there exists a function $f : h(i, S, q) \to h'(i, S, q)$ with $p \leq f(p)$ for all $p \in h(i, S, q)$.
This defines a preorder on the set of history functions of $\gst{G}$.
We write $h < h'$ if $h \leq h'$ and $h' \not\leq h$.
A history function $h$ is \emph{minimal} if there exists no history function $h'$ with $h' < h$.
Intuitively, a minimal history function is one where the histories of states stay in the early cycles of $\mathcal{A}$ as long as possible.
Since the choices of $h(i, S)$ and $h(i, S, q)$ can be made independently, minimal history functions always exist.

\begin{lemma}
  \label{lem:ChangeHistory}
  Let $\gst{G} \in 2^{2^Q}$ be any game state of $W(\mathcal{A})$.
  Then there exist $k \leq \lcm(k_1, \ldots, k_p) + 2 n + \max_{x \neq y} \lcm (k_x, k_y)$ and $m \leq \lcm(k_1, \ldots, k_p)$ such that $\delta_W(\gst{G}, A^k) \sim \delta_W(\gst{G}, A^{k+m})$.
\end{lemma}

The idea of the proof is that under a minimal history function, no state $q \in S \in \gst{G}$ can spend too long in a cycle it did not start in: the maximal number of steps is comparable to the $\lcm$ of the lengths of successive cycles.

\begin{proof}
  Denote the cycles of $\mathcal{A}$ by $C_1, \ldots, C_p$, so that $|C_i| = k_i$ for each $i$.
  Let $h$ be a minimal history function of $\gst{G}$.
  Define $\gst{G}_i$ for $i \geq 0$ as in Definition~\ref{def:History}.

  Let $t \geq 0$, $S \in \gst{G}_t$ and $q \in S$ be arbitrary, and let $S_0, \ldots, S_t = S$ and $q_0, \ldots, q_t = q$ be their histories under $h$.
  The history of $q$ travels through some of the cycles of $\mathcal{A}$, never entering the same cycle twice.
  We split the sequence $q_0, \ldots, q_n$ into words over $Q$ as $u_0 v_1^{p_1} u_1 v_2^{p_2} u_2 \cdots v_r^{p_r} u_r$, where
  \begin{itemize}
  \item
    each $p_j \geq 1$,
  \item
    each $v_j$ consists of the states of some cycle, which we may assume is $C_j$, repeated exactly once,
  \item
    the $u_j$ do not repeat states and each $u_j$ does not contain any states from $C_{j+1}$.
  \end{itemize}
  Intuitively, $v_j$ represents a phase of the history where the state stays in a cycle for several loops, and the $u_j$ represent transitions from one loop to another.
  Each $u_j$ ends right before the time step when the history of $q$ enters the loop $C_{j+1}$.
  It may share a nonempty prefix with $v_j$.
  See Figure~\ref{fig:HistoryLoop}(a).

    \begin{figure}[htp]
    \begin{center}
      \begin{tikzpicture}
        [nonaccept/.style={circle,draw,inner sep=0cm,minimum size=0.5cm},
        every edge/.style={draw,->,>=stealth},
        scale=0.8]

                \node [nonaccept] (a) at (-3.5,1) {};

        \node [nonaccept] (c1) at (0:2cm) {};
        \node [nonaccept] (c2) at (40:2cm) {};
        \node [nonaccept] (c3) at (80:2cm) {};
        \node [nonaccept] (c4) at (120:2cm) {};
        \node [nonaccept] (c5) at (160:2cm) {};
        \node [nonaccept] (c6) at (200:2cm) {};
        \node [nonaccept] (c7) at (240:2cm) {};
        \node [nonaccept] (c8) at (280:2cm) {};
        \node [nonaccept] (c9) at (320:2cm) {};

        \draw (a) edge (c5);

        \draw (c9) edge (c8);
        \draw (c8) edge (c7);
        \draw (c7) edge (c6);
        \draw (c6) edge (c5);
        \draw (c5) edge (c4);
        \draw (c4) edge (c3);
        \draw (c3) edge (c2);
        \draw (c2) edge (c1);
        \draw (c1) edge (c9);

        \begin{scope}[xshift=7cm]
          \node [nonaccept] (d1) at (0:1.5cm) {};
          \node [nonaccept] (d2) at (60:1.5cm) {};
          \node [nonaccept] (d3) at (120:1.5cm) {};
          \node [nonaccept] (d4) at (180:1.5cm) {};
          \node [nonaccept] (d5) at (240:1.5cm) {};
          \node [nonaccept] (d6) at (300:1.5cm) {};
        \end{scope}

        \draw (d6) edge (d5);
        \draw (d5) edge (d4);
        \draw (d4) edge (d3);
        \draw (d3) edge (d2);
        \draw (d2) edge (d1);
        \draw (d1) edge (d6);

        \node [nonaccept] (e1) at (3,1.5) {};
        \node [nonaccept] (e2) at (4.5,1.5) {};

        \draw (c2) edge (e1);
        \draw (e1) edge (e2);
        \draw (e2) edge (d3);

        \node [nonaccept] (b) at (9.5,1.5) {};

        \draw (d2) edge (b);

        \draw [thick,dashed,->] (165:1.5cm) arc (165:-165:1.5cm) node [pos=0.25,below] {$v_{i-1}$};
        \node at (0,0) {$C_{i-1}$};
        \draw [thick,dashed,->] (155:2.5cm) arc (155:45:2.5cm) arc (105:83:8cm) node [pos=0.5,above] {$u_{i-1}$};

        \begin{scope}[xshift=7cm]
          \draw [thick,dashed,->] (125:1cm) arc (125:-195:1cm) node [pos=0.15,below] {$v_i$};
          \node at (0,0) {$C_i$};
        \end{scope}

        \begin{scope}[xshift=-4.5cm,yshift=-5cm]
          \node at (0,7) {(a)};
          \node at (0,2) {(b)};
          
          \def\h{-3}

          \foreach \x in {
            0,1,
            3,4,5.5,6.75,
            9.25,10.5,11.75,
            13.75,15}{
            \draw (\x,-0.2) -- ++(0,0.4);
          }
          \draw (0,0) -- (1,0);
          \draw (3,0) -- (6.75,0);
          \draw (9.25,0) -- (11.75,0);
          \draw (13.75,0) -- (15,0);

          \foreach \x in {
            0,1,
            3,4,5,
            8,9,10.5,11.75,
            13.75,15}{
            \draw (\x,\h-0.2) -- ++(0,0.4);
          }
          \draw (0,\h) -- (1,\h);
          \draw (3,\h) -- (5,\h);
          \draw (8,\h) -- (11.75,\h);
          \draw (13.75,\h) -- (15,\h);

          \foreach \x/\y in {0/0, 3/0, 0/\h, 3/\h, 4/\h, 8/\h}{
            \node [below] at (\x+0.5,\y) {$v_{i-1}$};
          }
          \foreach \x/\y in {4/0, 9/\h}{
            \node [below] at (\x+0.75,\y) {$u_{i-1}$};
          }
          \foreach \x/\y in {5.5/0, 9.25/0, 10.5/0, 13.75/0, 10.5/\h, 13.75/\h}{
            \node [below] at (\x+0.625,\y) {$v_i$};
          }
          \foreach \x/\y in {2/0, 8/0, 12.75/0, 2/\h, 6.5/\h, 12.75/\h}{
            \node at (\x,\y) {$\cdots$};
          }

          \draw [decorate,decoration={brace,amplitude=10pt}] (0, 0.5) -- (4, 0.5) node [midway,above=13pt] {$p_{i-1}$};
          \draw [decorate,decoration={brace,amplitude=10pt}] (5.5, 0.5) -- (15, 0.5) node [midway,above=13pt] {$p_i$};

          \draw [decorate,decoration={brace,amplitude=10pt,aspect=0.75}] (0, \h+0.5) -- (9, \h+0.5) node [pos=0.75,above=13pt] {$p_{i-1}+a$};
          \draw [decorate,decoration={brace,amplitude=10pt}] (10.5, \h+0.5) -- (15, \h+0.5) node [midway,above=13pt] {$p_i-b$};

          \draw [dashed] (4, -0.2) -- (4,\h+0.2);
          \draw [dashed] (10.5, -0.2) -- (10.5,\h+0.2);

          \node [below] at (4,\h-0.5) {$s$};
          \node [below] at (10.5,\h-0.5) {$s+K$};
        \end{scope}

      \end{tikzpicture}
    \end{center}
    \caption{(a) Definition of the words $v_{i-1}$, $u_{i-1}$ and $v_i$. (b) A part of the history $h({\cdot}, S, q)$ (above) and $h'({\cdot}, S, q)$ (below). The braces count the number of repetitions of $v_{i-1}$ and $v_i$.}
    \label{fig:HistoryLoop}
  \end{figure}
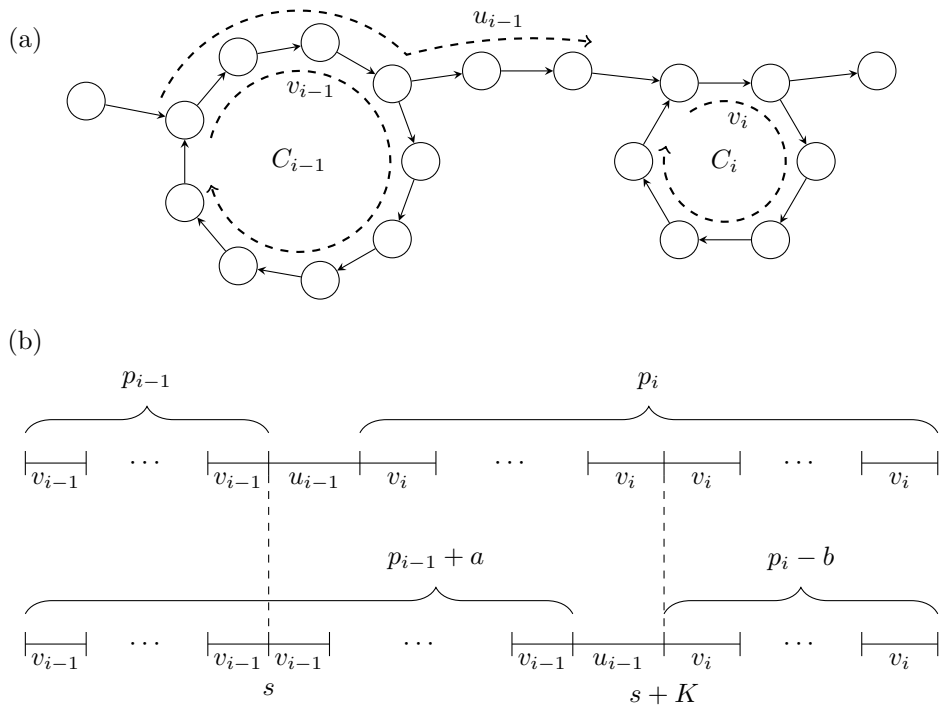

  We claim that $p_i k_i \leq \max_{x \neq y} \lcm (k_x, k_y)$ holds for all $1 < i \leq r$.
  Assume the contrary.
  Since $k_i = |v_i|$, we have in particular $p_i |v_i| > \lcm(|v_{i-1}|, |v_i|)$ for some $i$, so that $a |v_{i-1}| = b |v_i|$ holds for some $a > 0$ and $0 < b \leq p_i$.
  Denote $s = |u_0 v_1^{p_1} \cdots u_{i-2}|$, which is the time step after which the history of $q$ enters the repetitive portion of the previous loop $C_{i-2}$.
  Denote $K = |v_{i-1}^{p_{i-1}+a} u_{i-1}|$ and $q'_{s+1}, q'_{s+2}, \ldots, q'_{s + K} = v_{i-1}^{p_{i-1}+a} u_{i-1}$.
  Note that $q_{s+1}, q_{s+2}, \ldots, q_{s + K} = v_{i-1}^{p_{i-1}} u_{i-1} v_i^b$, so that we have $q'_{s+j} \leq q_{s+j}$ for all $1 \leq j < K$ and $q'_{s+K} < q_{s+K}$.
  See Figure~\ref{fig:HistoryLoop}(b).
  Our goal is to modify the history function $h$ by replacing each $q_{s+i}$ with $q'_{s+i}$ in this range.
  
  In $\mathcal{A}$ we have transitions from $q_s$ to both $q_{s+1}$ and $q'_{s+1}$, and from each $q'_{s+j}$ to $q'_{s+j+1}$, as well as from $q'_{s+K}$ to $q_{s+K+1}$.
  For $q \leq j \leq K$ the game state $\delta_W(\gst{G}, A^{s+j})$ contains $S'_{s+j} =: (S_{s+j} \setminus \{q_{s+j}\}) \cup \{q'_{s+j}\}$.
  We also have $S_{s+K+1} \in \delta_W(\{S'_{s+K}\}, A)$.
  There are now two possibilities.
  If $q'_{s+K} \in S_{s+K}$, then $S'_{s+K} \in \delta_W(\gst{G}, A^{s+K})$ is a proper subset of $S_{s+K}$, which contradicts our choice of $\gst{G}_{s+K}$ as a version of $\delta_W(\gst{G}, A^{s+K})$ with all proper supersets removed.
  If $q'_{s+K} \notin S_{s+K}$, then we may define a new history function $h'$ by defining $h'(s+K+1, S) = S'_{a+K}$, $h'(s+K+1, S, q) = (h(s+K+1, S, q) \setminus \{q_{s+K}\}) \cup \{q'_{s+K}\}$, and $h'(t, S', q') = h(t, S', q')$ for all other choices of $t$, $S'$ and $q'$.
  Then the function $f : h'(s+K+1, S, q) \to h(s+K+1, S, q)$ defined by $f(q'_{s+K}) = q_{s+K}$ and $f(q') = q'$ for other $q' \in h'(s+K+1, S, q)$ shows $h' < h$, which contradicts the local minimality of $h$.
  We have now shown $p_i k_i \leq \max_{x \neq y} \lcm (k_x, k_y)$ for all $1 < i \leq r$.

  Denote $L = \lcm(k_1, \ldots, k_p)$.
  If $t \geq L + 2 \ell + p \cdot \max_{x \neq y} \lcm (k_x, k_y)$, then $p_1 k_1 \geq L + \ell$, which implies $q_\ell = q_{\ell+L}$ (note that $|u_0| \leq \ell$, so that $q_\ell, q_{\ell + L} \in C_1$).
  Since this holds for every history of every state of $S$ under $h$ and each state of each set $S_i$ for $i \leq t$ can be chosen as $q_i$ for some $q \in S$, we have $S_\ell = S_{\ell+L}$.
  Then $S \in \delta_W(\{S_{\ell+L}\}, A^{t-\ell-L}) = \delta_W(\{S_\ell\}, A^{t-\ell-L})$, so in particular $S \in \delta_W(\gst{G}, A^{t-L}) \sim \gst{G}_{t-L}$.
  On the other hand, $S \in \delta_W(\{S_\ell\}, A^{t-\ell}) = \delta_W(\{S_{\ell+L}\}, A^{t-\ell})$, so $S \in \delta_W(\gst{G}, A^{t+L}) \sim \gst{G}_{t+L}$.
  Since $S \in \gst{G}_t$ was arbitrary, we have $\gst{G}_t \sim \gst{G}' \subseteq \gst{G}_{t-L}$ and $\gst{G}_t \sim \gst{G}'' \subseteq \gst{G}_{t+L}$ for some game states $\gst{G}', \gst{G}'' \in 2^{2^Q}$.
  By considering $t+L$ instead of $t$ and doing the same analysis, we obtain $\gst{G}_t \sim \gst{G}_{t+L}$.
\end{proof}

\begin{theorem}
  Let $\mathcal{A}$ be an $n$-state binary DFA whose language is bounded.
  Then there is a partition $\ell + k_1 + \cdots + k_p = n$ such that the minimal DFA for $W(\mathcal{L}(\mathcal{A}))$ has at most $\sum_{m = 0}^{\ell + p + 1} (p \cdot \max_{x \neq y} \lcm(k_x, k_y) + 2 \ell + 2 \lcm(k_1, \ldots, k_p))^m$ states.
\end{theorem}

\begin{proof}
  Denote the minimal DFA for $W(\mathcal{L}(\mathcal{A}))$ by $\mathcal{B}$.
  We may assume that $\mathcal{A}$ is minimal, and then it has disjoint cycles, as otherwise the number of length-$n$ words in $\mathcal{L}(\mathcal{A})$ would grow exponentially while in a bounded language this growth is at most polynomial.
  Let $k_1, \ldots, k_p$ be the lengths of the cycles and $\ell$ the number of remaining states, and denote $P = p \cdot \max_{x \neq y} \lcm(k_x, k_y) + 2 \ell + 2 \lcm(k_1, \ldots, k_p)$.
  Then any $w \in \mathcal{L}(W(\mathcal{A}))$ has $|w|_B \leq \ell+p$, as otherwise Bob can win by choosing to leave a cycle whenever possible.

  Consider a word $w = A^{t_0} B A^{t_1} B \cdots B A^{t_m}$ with $0 \leq m \leq \ell + p$.
  If $t_i \geq P$ for some $i$, then Lemma~\ref{lem:ChangeHistory} implies $\delta_W(\gst{G}, A^{t_i}) \sim \delta_W(\gst{G}, A^t)$ for the game state $\gst{G} = \delta_W(\{\{q_0\}\}, A^{t_0} B \cdots A^{t_{i-1}} B)$ and some $t < t_i$.
  Thus the number of distinct states of $\mathcal{B}$ reachable by such words is at most $P^{m+1}$.
  The claim follows.
\end{proof}

The upper bound we obtain (the maximum of the expression taken over all partitions of $n$) is at least $n^n$.
We do not know whether the actual complexity is superexponential for bounded languages.
If we combine the gadgets $\GenSubset_n$ and $\Testing_n$, the resulting DFA recognizes a language whose winning set requires at least $2^n$ states,
so for finite (and thus bounded) regular languages the state complexity of the winning set is at least exponential.

\section{Chain-Like Automata}

In this section we investigate a family of binary automata consisting of a chain of states with a self-loop on each state.
More formally, define a \emph{$1$-bounded chain DFA} as $\mathcal{A} = (Q, \{0,1\}, q_0, \delta, F)$ where $Q = \{0, 1, \ldots, n-1\}$, $q_0 = 0$, $\delta(i, 0) = i$ and $\delta(i, 1) = i+1$ for all $i \in Q$ except $\delta(n-1, 1) = n-1$, and $n-1 \notin F$.
See Figure~\ref{chainGeneral} for an example.
It is easy to see that these automata recognize exactly the regular languages $L$ such that $w \in L$ depends only on $|w|_1$, and $|w|_1$ is also bounded.
Of course, the labels of the transitions have no effect on the winning set $W(\mathcal{L}(\mathcal{A}))$ so the results of this section apply to every DFA with the structure of a $1$-bounded chain DFA.

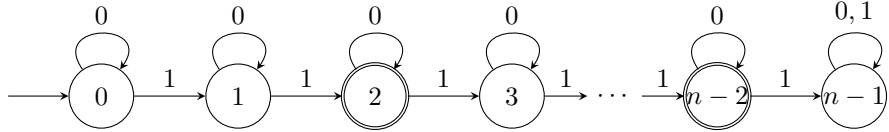
\begin{figure}[htp]
  \begin{center}
    \begin{tikzpicture}
	[nonaccept/.style={circle,draw,inner sep=0cm,minimum size=0.85cm},
	 accept/.style={circle,draw,double,inner sep=0cm,minimum size=0.85cm},
	 every edge/.style={draw,->,>=stealth},
	 scale=0.45]
		\node (i) at (-10, 0) {};
		\node [nonaccept] (n0) at (-7, 0) {$0$};
		\node [nonaccept] (n1) at (-3, 0) {$1$};
		\node [accept] (n2) at (1, 0) {$2$};
		\node [nonaccept] (n3) at (5, 0) {$3$};
		\node (n4) at (8, 0) {$\cdots$};
		\node [accept] (n5) at (11, 0) {$n-2$};
		\node [nonaccept] (n6) at (15, 0) {$n-1$};

		\draw (i) edge (n0);
		\foreach \k/\kk in {0/1,1/2,2/3,3/4,4/5,5/6}{
			 \draw (n\k) edge node [midway,above] {$1$} (n\kk);
		}
		\foreach \k in {0,1,2,3,5}{
			\draw [in=55, out=125, loop, looseness=4] (n\k) edge node [midway,above] {$0$} (n\k);
		}
		\draw [in=55, out=125, loop, looseness=4] (n6) edge node [midway,above] {$0,1$} (n6);
		
\end{tikzpicture}

  \end{center}

  \caption{\label{chainGeneral} A $1$-bounded chain DFA. Any states except $n-1$ can be final.}
\end{figure}

\begin{lemma}
  \label{chainEquivWords}
  Let $\mathcal{A}$ be an $n$-state $1$-bounded chain DFA. Let $\equiv$ stand for $\equiv_{W(\mathcal{L}(\mathcal{A}))}$.
  \begin{roster}
  \item
    \label{chain1} For every state $q \in Q$ and every $S \in \delta_W(\{\{q\}\}, A B)$, there exists $R \in \delta_W(\{\{q\}\}, B A)$ with $R \subseteq S$.
  \item
    \label{chain2} For all $k \in \mathbb{N}$, $B^k A^k B^{k+1} \equiv B^{k+1} A^k B^k$. 
  \item
    \label{chain3} For all $k \in \mathbb{N}$, $A^{k+1} B^k A^k \equiv A^k B^k A^{k+1}$. 
  \item
    \label{chain4} $A^{n-1} \equiv A^n$ and $B^{n-1} \equiv B^n$.
  \end{roster}
\end{lemma}

The intuition for~\ref{chain1} is that the turn order $B A$ is better for Alice than $A B$, since she can undo any damage Bob just caused.
The other items are proved by concretely analyzing the evolution of game states, which $A$ intuitively ``moves around with precise control'' and $B$ ``thickens''. Lemma~\ref{lemmasEquivIncl} simplifies the analysis.

\begin{proof}
  Label the states of $\mathcal{A}$ by $\{0, 1, \ldots, m-1\}$ as in the definition of $1$-bounded chain DFA.
  For \ref{chain1}, \ref{chain2} and \ref{chain3} we ``unroll'' the self-loop at $m-1$ labeled by $1$ to obtain
  an equivalent automaton with an infinite chain of states, simplifying the arguments.
  We also argue in terms of the NFA for $W(\mathcal{L}(\mathcal{A}))$, saying that ``we produce a set $R \subseteq Q$ from $S \subseteq Q$ by reading $w \in \{A,B\}^*$'' if $R \in \delta_W(\{S\}, w)$.
  
  In this formalism \ref{chain1} means that for any set produced from $\{q\}$ by reading $A B$, we can produce a subset of it by reading $B A$.
  To see what sets can be produced, we use
  a spacetime diagram (Figure~\ref{ABinclBA}) where the time increases to the south. In this 
  diagram, by reading a $B$, a selected state will spread south and southeast. By reading an
  $A$, we have both possibilities, resulting in multiple sets.
  The claim follows directly.

  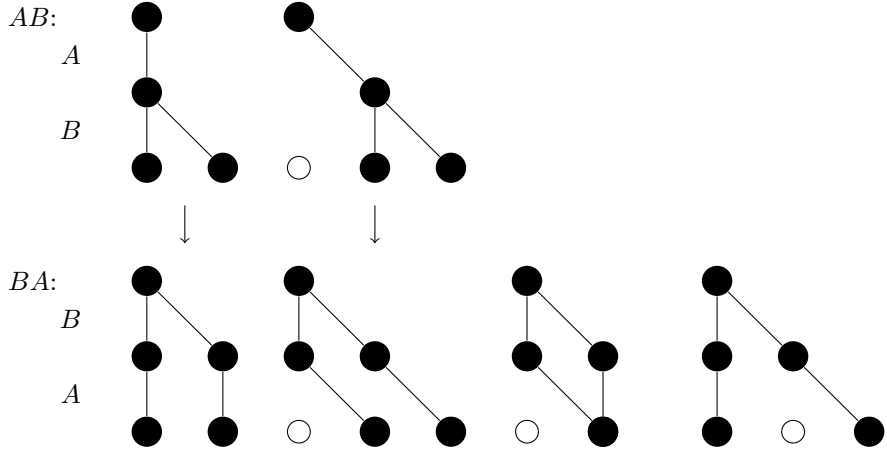
\begin{figure}[htp]
    
    \begin{center}
      \begin{tikzpicture}[scale=0.5,
        dot/.style={circle,fill,minimum size=0.4cm,inner sep=0cm},
        small/.style={circle,draw,minimum size=0.3cm,inner sep=0cm}]
		\node (0) at (-4, 4) {$AB$:};
		\node (1) at (-4, -3) {$BA$:};
		\node [dot] (2) at (-1, 4) {};
		\node [dot] (3) at (-1, 2) {};
		\node [dot] (4) at (-1, 0) {};
		\node [dot] (5) at (1, 0) {};
		\node [dot] (6) at (3, 4) {};
		\node [dot] (7) at (5, 2) {};
		\node [dot] (8) at (5, 0) {};
		\node [dot] (9) at (7, 0) {};
		\node [dot] (10) at (-1, -3) {};
		\node [dot] (11) at (-1, -5) {};
		\node [dot] (12) at (-1, -7) {};
		\node [dot] (13) at (1, -7) {};
		\node [dot] (14) at (14, -3) {};
		\node [dot] (15) at (16, -5) {};
		\node [dot] (17) at (18, -7) {};
		\node [dot] (18) at (1, -5) {};
		\node [dot] (19) at (14, -5) {};
		\node [dot] (20) at (14, -7) {};
		\node [dot] (21) at (3, -3) {};
		\node [dot] (22) at (3, -5) {};
		\node [dot] (23) at (5, -7) {};
		\node [dot] (24) at (7, -7) {};
		\node [dot] (25) at (5, -5) {};
		\node [dot] (26) at (9, -3) {};
		\node [dot] (27) at (11, -5) {};
		\node [dot] (28) at (11, -7) {};
		\node [dot] (29) at (9, -5) {};
		\node [dot] (30) at (11, -7) {};
		\node (35) at (0, -2) {};
		\node (36) at (0, -1) {};
		\node (37) at (5, -1) {};
		\node (38) at (5, -2) {};
		\node (39) at (-3, 3) {$A$};
		\node (40) at (-3, 1) {$B$};
		\node (41) at (-3, -4) {$B$};
		\node (42) at (-3, -6) {$A$};
		\node [small] (43) at (16, -7) {};
		\node [small] (44) at (3, -7) {};
		\node [small] (45) at (3, 0) {};
		\node [small] (46) at (9, -7) {};
                
		\draw (2) to (3);
		\draw (3) to (5);
		\draw (3) to (4);
		\draw (6) to (7);
		\draw (7) to (8);
		\draw (7) to (9);
		\draw (10) to (11);
		\draw (11) to (12);
		\draw (14) to (15);
		\draw (15) to (17);
		\draw (10) to (18);
		\draw (18) to (13);
		\draw (14) to (19);
		\draw (19) to (20);
		\draw (21) to (22);
		\draw (22) to (23);
		\draw (21) to (25);
		\draw (25) to (24);
		\draw (26) to (27);
		\draw (27) to (28);
		\draw (26) to (29);
		\draw (29) to (30);
		\draw [->] (36.center) to (35.center);
		\draw [->] (37.center) to (38.center);
\end{tikzpicture}

    \end{center}

    \caption{\label{ABinclBA} Sets produced with $AB$ and $BA$. The arrows indicate a subset relation.}
  \end{figure}

  We now prove \ref{chain2}.
  By Lemma~\ref{lemmasEquivIncl}\ref{equivChainSingletons} it is enough to consider singleton states $\{q\}$, and without loss of generality we assume $q = 0$.
  By item~\ref{chain1} of this result and Lemma~\ref{lemmasEquivIncl}\ref{equivChainDoubleIncl} it is sufficient to prove that for every set obtained
  from $B^{k+1} A^k B^k$, we can produce a subset of it by reading $B^k A^k B^{k+1}$.
  After reading $B^{k+1}$, we have the interval $\{0, \dots k+1\}$. After $A^k$ 
  we get a set included in $\{0,\dots, 2k+2\}$ where the distance between every two consecutive elements is less than $k$.
  After reading $B^k$ the gaps are filled and we get an interval containing
  $ \{k, \dots, 2k+2\}$.
  When reading $B^k A^k B^{k+1}$ we do the following: get $\{0, \dots k\}$ with $B^k$,
  make every element go to position $k$ with $A^k$, and extend with $B^{k+1}$ to get the interval $\{k, \dots, 2k+2\}$.

  For \ref{chain3}, for the same reason as previously it is sufficient to prove that for every set obtained from $\{0\}$
  by reading $A^k B^k A^{k+1}$, we can produce a subset by reading $A^{k+1} B^k A^k$.
  First we prove that by reading $A^{k+1} B^k A^k$ we can get any singleton set
  $\{k\}, \{k+1\}, \ldots, \{2k+1\}$: By reading $A^{k+1} B^k$, we can get any singleton
  set between $0$ and $k+1$ and expand it to have any interval of length $k+1$
  between $0$ and $2k+1$.
  Then by reading $A^k$ we can have a singleton state at the end position
  of the interval, that is between $k$ and $2k+1$.
  See Figure~\ref{AnABnAn}.

  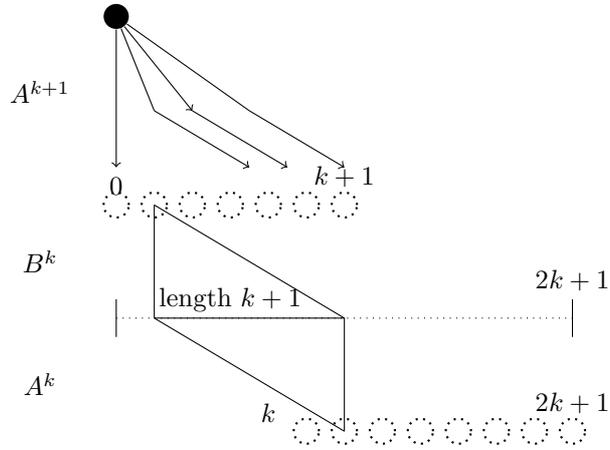
\begin{figure}[htp]
    
    \begin{center}
      \begin{tikzpicture}[scale=0.5,
        dot/.style={circle,fill,minimum size=0.3cm},
        dottednode/.style={circle,draw,thick,dotted,minimum size=0.3cm},]
		\node [dot] (48) at (0, 13) {};
		\node [dottednode] (49) at (0, 8) {};
		\node [dottednode] (50) at (6, 8) {};
		\node [dottednode] (51) at (1, 8) {};
		\node [dottednode] (52) at (2, 8) {};
		\node [dottednode] (53) at (3, 8) {};
		\node [dottednode] (54) at (4, 8) {};
		\node [dottednode] (55) at (5, 8) {};
		\node (56) at (0, 9) {};
		\node (57) at (2, 10.5) {};
		\node (58) at (4.5, 9) {};
		\node (59) at (6, 8.75) {$k+1$};
		\node (60) at (0, 8.5) {$0$};
		\node (61) at (1, 5) {};
		\node (62) at (6, 5) {};
		\node (63) at (1, 8) {};
		\node (64) at (0, 5) {};
		\node (65) at (12, 5) {};
		\node (66) at (12, 6) {$2k + 1$};
		\node (67) at (6, 2) {};
		\node (68) at (-2, 11) {$A^{k+1}$};
		\node (69) at (-2, 6.5) {$B^k$};
		\node (70) at (-2, 3.25) {$A^k$};
		\node (73) at (4, 2.5) {$k$};
		\node (74) at (12, 2.75) {$2k + 1$};
		\node (75) at (1, 10.5) {};
		\node (76) at (3.5, 9) {};
		\node (77) at (3.5, 10.5) {};
		\node (78) at (6, 9) {};
		\node (79) at (3, 5.5) {length $k+1$};
		\node [dottednode] (80) at (5, 2) {};
		\node [dottednode] (81) at (11, 2) {};
		\node [dottednode] (82) at (6, 2) {};
		\node [dottednode] (83) at (7, 2) {};
		\node [dottednode] (84) at (8, 2) {};
		\node [dottednode] (85) at (9, 2) {};
		\node [dottednode] (86) at (10, 2) {};
		\node [dottednode] (87) at (12, 2) {};
		\node (88) at (0, 5.5) {};
		\node (89) at (0, 4.5) {};
		\node (90) at (12, 5.5) {};
		\node (91) at (12, 4.5) {};
                
		\draw [->] (48) to (56.center);
		\draw [->] (48) to (57.center);
		\draw [->] (57.center) to (58.center);
		\draw (63.center) to (61.center);
		\draw (61.center) to (62.center);
		\draw (62.center) to (63.center);
		\draw [dotted] (64.center) to (65.center);
		\draw (61.center) to (67.center);
		\draw (67.center) to (62.center);
		\draw [->] (75.center) to (76.center);
		\draw [->] (77.center) to (78.center);
		\draw (77.center) to (48);
		\draw (48) to (75.center);
		\draw (88.center) to (89.center);
		\draw (90.center) to (91.center);
\end{tikzpicture}

    \end{center}

    \caption{\label{AnABnAn} The sets produced by $A^{k+1} B^k A^k$.}
  \end{figure}

  Now we prove that every set obtained by reading  $A^k B^k A^{k+1}$ has 
  at least one of $k, \ldots, 2k+1$.
  After reading $A^k B^k$ we have any interval of size $k+1$ between $0$ and $2k$. The state $k$ is always in the interval,
  and must go somewhere between positions $k$ and $2k+1$ after $A^{k+1}$.

  %
  %
  %
  %

  As for \ref{chain4}, reading $A^{n-1}$ or $A^n$ in any singleton game state $\{q\}$ produces exactly the states $\{q\}, \{q+1\}, \ldots, \{m-1\}$ and the loop $\{m\}, \ldots, \{m+p-1\}$.
  Likewise, reading $B^{n-1}$ or $B^n$ produces $\{q, q+1, \ldots, m-1\} \cup \{m, \ldots, m+p-1\}$.
\end{proof}

\begin{theorem}
  \label{chainBoundedBound}
  Let $\mathcal{A}$ be a $1$-bounded chain DFA with $n$ states.
  The number of states in the minimal DFA of $W(\mathcal{L}(\mathcal{A}))$ is $O(n^{1/5} e^{4 \pi \sqrt{\frac{n}{3}}})$.
\end{theorem}
\begin{proof}
  Since $\mathcal{A}$ does not accept any word with $n$ or more $1$-symbols, $W(\mathcal{L}(\mathcal{A}))$ contains no word with $n$ or more $B$-symbols.
  Lemma~\ref{chainEquivWords}\ref{chain2} and~\ref{chain3}
  allow us to rewrite every word of $W(\mathcal{L}(\mathcal{A}))$ in the form $A^{n_1} B^{n_2} A^{n_3} B^{n_4} \cdots A^{n_{2r-1}} B^{n_{2r}}$ where
  the sequence $n_1, \ldots, n_{2r}$ is first nondecreasing and then nonincreasing, and $n_2 + n_4 + \cdots + n_{2r} < n$.
  With Lemma~\ref{chainEquivWords}\ref{chain4} we can also guarantee $n_1, n_3, \ldots, n_{2r-1} < n$, so that $\sum_i n_i < 4 n$.
  In \cite{auluck1951some}, Auluck showed that the number $Q(m)$ of partitions $m = n_1 + \ldots n_r$ of an integer $m$ that are first nondecreasing and then nonincreasing is $\Theta(m^{-4/5} e^{2 \pi \sqrt{m/3}})$.
  Of course, $v \equiv w$ implies $v \sim w$.
  Thus the number of non-right-equivalent words for $W(\mathcal{L}(\mathcal{A}))$, and the number of states in its minimal DFA, is at most $1 + \sum_{m = 0}^{4 n-1} Q(m) = O(n^{1/5} e^{4 \pi \sqrt{\frac{n}{3}}})$.
\end{proof}

\section{Case Study: Exact Number of $1$-Symbols}

In the previous section we bounded the complexity of the winning set of certain bounded permutation invariant languages. Here we study 
a particular case, the language of words with exactly $n$ ones, or $L= (0^*1)^n 0^*$.
We not only compute the number of states in the minimal automaton (which is cubic in $n$), but also describe the winning set.
Throughout the section $\mathcal{A}$ is the minimal automaton for $L$, described in Figure \ref{01s01s}.
For $S \subseteq Q$, we denote $\overline{S} =  \{\min(S), \min(S) + 1, \ldots, \max(S)\}$, and for any game state $\gst{G}$ of $W(\mathcal{A})$, denote $\overline{\gst{G}} = \{ \overline{S} \;|\; S \in \gst{G}\}$.

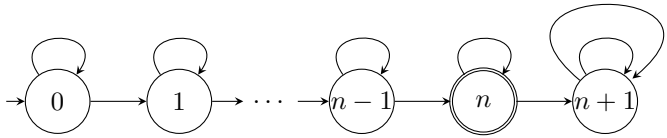
\begin{figure}[htp]
  
  \begin{center}
    \begin{tikzpicture}
	[nonaccept/.style={circle,draw,inner sep=0cm,minimum size=0.85cm},
	 accept/.style={circle,draw,double,inner sep=0cm,minimum size=0.85cm},
	 every edge/.style={draw,->,>=stealth},
	 scale=0.4]
		\node (0) at (-9, 0) {};
		\node [nonaccept] (5) at (-7, 0) {$0$};
		\node [nonaccept] (6) at (-3, 0) {$1$};
		\node (7) at (0, 0) {$\dots$};
		\node [nonaccept] (8) at (3, 0) {$n-1$};
		\node [accept] (9) at (7, 0) {$n$};
		\node [nonaccept] (11) at (11, 0) {$n+1$};
		\draw (0) edge (5);
		\draw (5) edge (6);
		\draw (6) edge (7);
		\draw (7) edge (8);
		\foreach \k in {5,6,8,9,11}{
			 \draw [in=55, out=125, loop, looseness=4] (\k) edge (\k);
		}
		\draw [in=40, out=140, loop] (11) edge (11);
		\draw (8) edge (9);
		\draw (9) edge (11);
\end{tikzpicture}

  \end{center}
  
  \caption{\label{01s01s} The minimal DFA for $L= (0^*1)^n0^*$.}
\end{figure}

\begin{lemma} \label{lem:stateIsInterval}
  Each game state $\gst{G}$ of $W(\mathcal{A})$ is equivalent to $\overline{\gst{G}}$.
\end{lemma}

The idea is that the left and right ends of sets in $\gst{G}$ evolve independently of their other elements, and their positions determine whether $\gst{G}$ is final.

\begin{proof}
  We prove by induction that for every
  $w\in \{A,B\}^*$, the game state $\delta_W(\gst{G},w)$ is final iff $\delta_W ( \overline{\gst{G}},w)$ is.
  \begin{itemize}
  \item  For the empty word $\lambda$, $\delta_W (\gst{G}, \lambda) = \gst{G}$ is final
    iff $\{n\} \in \gst{G}$
    iff $\{n\} \in \overline{\gst{G}} = \delta_W(\overline{\gst{G}}, \lambda)$.
  \item For $w= B v$,
    it's easy to see that $ \delta_W (\overline{\gst{G}}, B) = \overline{ \delta_W (\gst{G},B)}$.
    By the induction hypothesis, $ \delta_W (\gst{G}, B v) = \delta(\delta_W (\gst{G}, B), v)$ is final iff
    $ \delta_W (\overline{\delta_W (\gst{G}, B)}, v) = \delta_W (\delta_W (\overline{\gst{G}}, B), v) = \delta_W (\overline{\gst{G}}, B v)$ is.
  \item For $w= A v$,
    we have
    $ \overline{\delta_W (\overline{\gst{G}}, A)} = \overline{ \delta_W (\gst{G}, A)}$: 
    For each set $S \in \gst{G}$, we focus on a set $R$ that can be obtained
    by a combination of choices by reading one $A$.
    The set $\overline{R}$ is only defined by the leftmost and rightmost elements in $R$.
    From $\overline{S}$, by making the same choices for the leftmost and rightmost
    elements we can produce the same leftmost and rightmost elements as in $R$ to obtain a set equivalent to $\overline{R}$.

    By the induction hypothesis, the game state $\delta_W (\gst{G}, A v) = \delta_W (\delta_W (\gst{G}, A), v)$ is final iff
    $\delta_W (\overline{\delta_W (\gst{G}, A)}, v) = \delta_W (\overline{\delta_W (\overline{\gst{G}}, A)}, v)$ is.
    Again by the induction hypothesis, this is equivalent to $\delta_W (\delta_W (\overline{\gst{G}}, A), v) = \delta_W(\overline{\gst{G}}, A v)$ being final.
  \end{itemize}
\end{proof}

\begin{lemma} \label{lem:formOfReachables}
  Let $T$ be the set of integer triples $(i, \ell, N)$ with $0 \leq i \leq n$, $1 \leq \ell \leq n-i+1$ and $1 \leq N \leq n-i-\ell+2$.
  For $(i, \ell, N) \in T$, let
  \[ \gst{G}(i, \ell, N) :=
    \{
    \{i, \ldots, i+\ell-1\},
    \{i+1, \ldots, i+\ell\},
    \ldots,
    \{i+N-1, \ldots, i+\ell+N-2\} \}.
  \]
  \begin{roster}
  \item \label{form1}
    Each reachable game state of $W(\mathcal{A})$ is equivalent to some $\gst{G}(i, \ell, N)$ for $(i, \ell, N) \in T$, or to $\emptyset$.
  \item \label{form2} The game states $\gst{G}(i,\ell,N)$ for $(i, \ell, N) \in T$ are nonequivalent.
  \item \label{form3} Every $\gst{G}(i,\ell,N)$ for $(i, \ell, N) \in T$ is equivalent to some reachable game state.
  \end{roster}
\end{lemma}

The game state $\gst{G}(i,\ell,N)$ is an interval of intervals, where $i$ is the leftmost position of the first interval, $\ell$ is their common length, and $N$ is their number.
The first item is proved by induction, and the others by exhibiting a word over $\{A,B\}$ that produces or separates given game states.

\begin{proof}
  We first make the following remarks which follow from Lemma~\ref{lem:stateIsInterval}.
  Let $S = \{ i, \ldots, j\}$. If $i=j$, then $\delta_W (\{S\}, A) \sim \{ \{i\}, \{i+1\} \}$,
  otherwise $\delta_W (\{S\}, A) \sim \{ \{i+1, \dots, j\} \}$.
  In both cases we also have $\delta_W (\{S\}, B) \sim \{ \{i, \dots, j+1\} \}$.
  
  We prove~\ref{form1}.
  From $\gst{G}(i,\ell,N)$, if we read $B$, we have $\gst{G}(i,\ell+1,N)$. If we read $A$ and $\ell=1$, we get $\gst{G}(i,1,N+1)$. If we read $A$ and $\ell > 1$, we obtain $\gst{G}(i+1,\ell-1,N)$.
  The claim follows since the initial game state is $\gst{G}(0, 1, 1)$, and if a set in the game state contains the state $n+1$, it is equivalent to $\emptyset$.
  
  For~\ref{form2}, we first distinguish game states with different lengths $\ell<\ell'$.
  By reading $A^\ell (B A)^k$ for a suitable $k \geq 0$ we reach a final game state from $\gst{G}(i, \ell, N)$ but not from any $\gst{G}(i', \ell', N')$ with $\ell < \ell'$.
  
  Now we suppose we have game states $\gst{G}(i,\ell,N)$ and $\gst{G}(i',\ell,N')$.
  By reading $A^{\ell-1}$, we get respectively $\gst{G}(i+\ell-1,1,N)$
  and $\gst{G}(i'+\ell-1,1,N')$, which are intervals of singleton sets. Since 
  these distributions of singleton sets are different, we can read $(B A)^k$ 
  for a suitable $k \geq 0$ to obtain a final game state from one of them but not the other.
  
  For~\ref{form3}, reading $(B A)^i A^{N-1} B^{\ell-1}$ leads to a game state equivalent to $\gst{G}(i,\ell,N)$.
\end{proof}

In particular, the minimal DFA for $W(L)$ has $|T|+1$ states, and counting them yields the following.

\begin{proposition} \label{lem:formOfH}
  The minimal DFA for $W(L)$ has $\frac{n^3}{6}+n^2+\frac{11n}{6}+2$ states.
\end{proposition}


\begin{proposition}
  $W(L)$ is exactly the set of words $w \in \{A, B\}^*$ such that $|w|_A \geq n$, $|w|_B \leq n$, and every suffix $v$ of $w$ satisfies $|v|_A \geq |v|_B$.
\end{proposition}

\begin{proof}
  Every $w \in W(L)$ satisfies $|w|_A \geq n$ since Bob can choose to play only $0$s, and $|w|_B \leq n$ since he can choose to play only $1$s.
  To see that the suffix condition is necessary, note first that only game states of the form $\gst{G}(i,1,N)$ can be accepting.
  Since reading $A$ decreases the parameter $\ell$ by one,
  and $B$ increases $\ell$ by one, words of $W(L)$ must have after each $B$ an associated $A$ somewhere in the word.
  This is equivalent to the suffix condition.
  
  Conversely, on words of the given form Alice can win by associating to each $w_i = B$ some $w_j = A$ that occurs after it, choosing $v_j \neq v_i$ for the constructed word $v \in \{0,1\}^*$, and choosing the remaining symbols so that $|v|_1 = n$.
\end{proof}

\section{Membership and Intersection}

In this section we investigate the computational complexity of deciding membership in the winning set of a regular language, and nonemptiness of intersection with a given regular language over $\{A,B\}$.
The former turns out to be P-complete and the latter PSPACE-complete.
Our hardness results rely on simulation of boolean circuits.
Such simulation allows us to reduce P-hardness of circuit evaluation and PSPACE-hardness of the iterated circuit prediction to the aforementioned decision problems on the winning set.

For the purposes of this section, a \emph{boolean circuit} is a finite directed acyclic graph $C$.
There are five types of nodes in $C$: \textbf{AND}-gates and \textbf{OR}-gates with indegree 2 and arbitrary outdegree, \textbf{NOT}-gates of indegree 1 and arbitrary outdegree, input nodes with indegree 0 and outdegree 1, and output nodes with indegree 1 and outdegree 0.
If $C$ contains $k$ input and $m$ output nodes, it transforms elements of $\{T,F\}^k$ into those of $\{T,F\}^m$ by evaluation of the gates.
In the \emph{circuit value problem}, we are given a circuit $C$ with $k$ inputs and a single output, and a vector $a \in \{T, F\}^k$, and are asked whether $C(a) = T$; this problem is P-complete with respect to logarithmic space reductions~\cite[Theorem 6.30]{complexity}.
In the \emph{iterated circuit problem}, we are we are given a circuit $C$ with $k$ inputs and outputs, a vector $a \in \{T, F\}^k$ and an index $i \in \{0, \ldots, k-1\}$, and are asked whether $C^n(a)_1 = T$ for some $n \geq 0$; this problem is PSPACE-complete with respect to polynomial time reductions~\cite[Lemma 3]{goles2016pspace}.

The general idea is that each gate of the circuit corresponds to a gadget formed by a subset of states in the automaton, and the simulation is done ``backwards in time'' in the sense that the initial state of the winning set automaton corresponds to the output of the circuit.
Each output node of the circuit corresponds to a pair of states in the automaton, and an output vector is encoded by choosing one state from each pair as initial.
Then by reading an appropriate word, we can construct a game state corresponding to all possible inputs that map to the chosen output, and accept only if the chosen input is represented in this game state.
The following result codifies this idea.

\begin{lemma}\label{circuitSimulation}
  Let $C$ be a boolean circuit with $k$ inputs, $m$ outputs, and $n$ gates. Then there is a number $p = O(n)$ and an acyclic DFA $\mathcal{A}_C = (Q, \{0,1\}, q_0, \delta, F)$ with $O(n(m+k+n))$ states, and special states 
  $q^T_i, q^F_i \in Q$ for $i = 0, \ldots, k-1$ and
  $p^T_i, p^F_i \in Q$ for $i = 0, \ldots, m-1$
  such that the following holds in $W(\mathcal{A}_C)$:
  for all $a \in \{T,F\}^m$ we have $\delta_W (\{p^{a_0}_0, \ldots, p^{a_{m-1}}_{m-1}\} , (AAB)^p ) = \gst{G} \cup \gst{H}$, where
  $\gst{G} = \{ \{ q^{b_0}_0, \ldots, q^{b_{k-1}}_{k-1} \} \;|\; C(b_0, \dots, b_{k-1}) = a_0, \dots, a_{m-1} \}$ and
  every $X \in \gst{H}$ satisfies $\{q^T_i,q^F_i\} \subseteq X$ for some $i$.
\end{lemma}

Here, $\gst{G}$ encodes exactly the inputs of the circuit that produce the output $a$, while $\gst{H}$ contains artifacts of the construction that we can safely ignore later on.

\begin{proof}
  We implement gadgets for \textbf{OR}-gates and \textbf{NOT}-gates of outdegree 1, and \textbf{SPLIT}-gates which split one input into multiple copies.
  They are shown in Figure~\ref{fig:example}.
  Then \textbf{AND}-gates can be implemented by de Morgan's law, and gates with arbitrary outdegree by \textbf{SPLIT}-gates.
  The boolean values are represented in a dual-rail fashion, so a gadget corresponding to a gate $g$ with indegree $n_1$ and outdegree $n_2$ will have $2 n_1$ input states $s^b_i = s^b_i(g)$ for $b \in \{T,F\}$ and $i = 0, \ldots, n_1-1$, and $2 n_2$ output states $t^b_i = t^b_i(g)$ for $b \in \{T,F\}$ and $i = 0, \ldots, n_2-1$.
  A subset of the $q^b_i$ is \emph{consistent} if it contains exactly one of $q^T_i$ and $q^F_i$ for each $i$, and \emph{excessive} if it contains both $q^T_i$ and $q^F_i$ for some $i$.
  Similar definitions are given for the output states.
  The gates are constructed so that by reading $AAB$, a consistent subset $S$ of the output states is transformed into a set of subsets of the input variables whose consistent members are exactly the inputs corresponding to $S$, and the rest are excessive.
  
  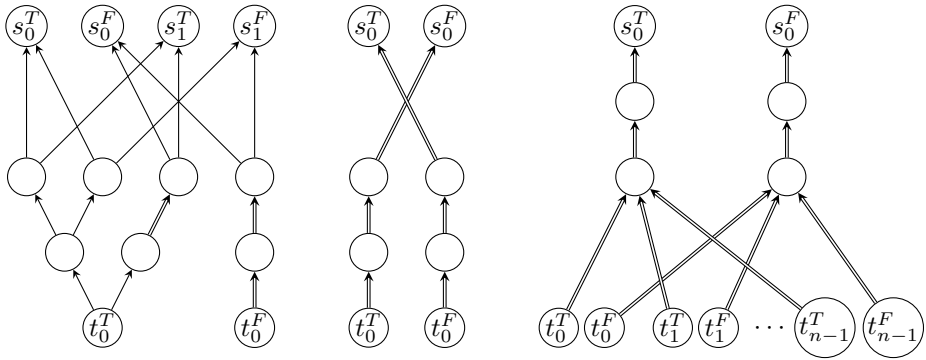
\begin{figure}[htp]%
    \begin{center}

      \begin{tikzpicture}
        [every node/.style={circle,draw,inner sep=0cm,minimum size=0.5cm},
        every edge/.style={draw,->,>=stealth}]

        \node (i1) at (0,0) {$t^T_0$};
        \node (i2) at (2,0) {$t^F_0$};

        \node (a1) at (-0.5,1) {};
        \node (a2) at (0.5,1) {};
        \node (a3) at (2,1) {};

        \draw (i1) edge (a1) edge (a2);
        \draw (i2) edge[double] (a3);

        \node (b1) at (-1,2) {};
        \node (b2) at (0,2) {};
        \node (b3) at (1,2) {};
        \node (b4) at (2,2) {};

        \draw (a1) edge (b1) edge (b2);
        \draw (a2) edge[double] (b3);
        \draw (a3) edge[double] (b4);

        \node (o1) at (-1,4) {$s^T_0$};
        \node (o2) at (0,4) {$s^F_0$};
        \node (o3) at (1,4) {$s^T_1$};
        \node (o4) at (2,4) {$s^F_1$};

        \draw (b1) edge (o1) edge (o3);
        \draw (b2) edge (o1) edge (o4);
        \draw (b3) edge (o2) edge (o3);
        \draw (b4) edge (o2) edge (o4);

        \begin{scope}[xshift=3.5cm]
          
          \node (i1) at (0,0) {$t^T_0$};
          \node (i2) at (1,0) {$t^F_0$};

          \node (a1) at (0,1) {};
          \node (a2) at (1,1) {};

          \draw (i1) edge[double] (a1);
          \draw (i2) edge[double] (a2);

          \node (b1) at (0,2) {};
          \node (b2) at (1,2) {};

          \draw (a1) edge[double] (b1);
          \draw (a2) edge[double] (b2);

          \node (o1) at (0,4) {$s^T_0$};
          \node (o2) at (1,4) {$s^F_0$};

          \draw (b1) edge[double] (o2);
          \draw (b2) edge[double] (o1);
        \end{scope}

        \begin{scope}[xshift=7cm]

          \node (i1) at (-1,0) {$t^T_0$};
          \node (i2) at (-0.4,0) {$t^F_0$};
          \node (i3) at (0.5,0) {$t^T_1$};
          \node (i4) at (1.1,0) {$t^F_1$};

          \node[draw=none] at (1.8,0) {$\cdots$};

          \node (i5) at (2.5,0) {$t^T_{n-1}$};
          \node (i6) at (3.4,0) {$t^F_{n-1}$};

          \node (a1) at (0,2) {};
          \node (a2) at (2,2) {};

          \draw (i1) edge[double] (a1);
          \draw (i2) edge[double] (a2);
          \draw (i3) edge[double] (a1);
          \draw (i4) edge[double] (a2);
          \draw (i5) edge[double] (a1);
          \draw (i6) edge[double] (a2);

          \node (b1) at (0,3) {};
          \node (b2) at (2,3) {};

          \draw (a1) edge[double] (b1);
          \draw (a2) edge[double] (b2);

          \node (o1) at (0,4) {$s^T_0$};
          \node (o2) at (2,4) {$s^F_0$};

          \draw (b1) edge[double] (o1);
          \draw (b2) edge[double] (o2);
        \end{scope}

      \end{tikzpicture}
    \end{center}
    
    \caption{Implementation of the gate gadgets: \textbf{OR}-gate on the left, \textbf{NOT}-gate in the middle, and \textbf{SPLIT}-gate of outdegree $n$ on the right.}
    \label{fig:example}%
  \end{figure}
  
  For the \textbf{OR}-gate, the set $S = \{t^T_1\}$ is transformed by $AA$ into three sets that are expanded with $B$ to obtain the sets $\{s_0^T, s_1^T\}$, $\{s_0^T, s_1^F\}$ and $\{s_0^F, s_1^T\}$.
  The set $S = \{t^F_0\}$ is transformed by $AAB$ into $\{s_0^F, s_1^F\}$.
  If $S = \{t^T_1, t^F_1\}$ is excessive, then so is the result.
  
  For the \textbf{NOT}-gate, we simply exchange the roles of the states corresponding to $T$ and $F$.
  This gadget also preserves excessive sets.
  The \textbf{SPLIT}-gate gadget transforms a set of the form $S = \{t_0^b, t_1^b, \ldots, t_{n-1}^b\}$ for $b \in \{T,F\}$ into $\{s_0^b\}$, and other consistent or excessive sets into excessive sets.

  Given a boolean circuit $C$, we assemble these gadgets into the DFA $\mathcal{A}_C$ in the following way, depicted in Figure~\ref{fig:circuit-layers}.
  First, we may assume that each edge on $C$ connects two gates of adjacent depth (minimal distance from an input node), and that each output node has the same depth $d(C)$.
  If this is not the case, we can enforce it by replacing some subset of edges by paths of \textbf{SPLIT}-gates of degree 1.
  For the input and output nodes of the circuit, we add the states $q^T_i, q^F_i \in Q$ for $i = 0, \ldots, k-1$ and $p^T_i, p^F_i \in Q$ for $i = 0, \ldots, m-1$ to the automaton.
  Next, we process each gate $g$ in increasing order of depth.
  Let $n_1, n_2 \geq 1$ be the indegree and outdegree of $g$ in $C$.
  We add to $Q$ copies of the gadget corresponding to $g$ and a gadget corresponding to a \textbf{SPLIT}-gate $S_g$ of degree $n_2$, and identify $t_0^b(g)$ with $s_0^b(S_g)$ for each $b \in \{T,F\}$.
  For each incoming edge $h \to g$, let $g$ be the $j$th out-neighbor of $h$ and $h$ the $\ell$th in-neighbor of $g$.
  Then we identify $t_j^b(S_h)$ with $s_\ell^b(g)$ for $b \in \{T,F\}$.
  If $d = d(C)$ is the depth of $C$, then we do not add the \textbf{SPLIT}-gate gadget $S_g$, but identify the input nodes $t_0^b(g)$ with $p_i^b$ instead, where $g$ is the $i$th output node of $C$.
  This concludes the construction of $\mathcal{A}_C$.

  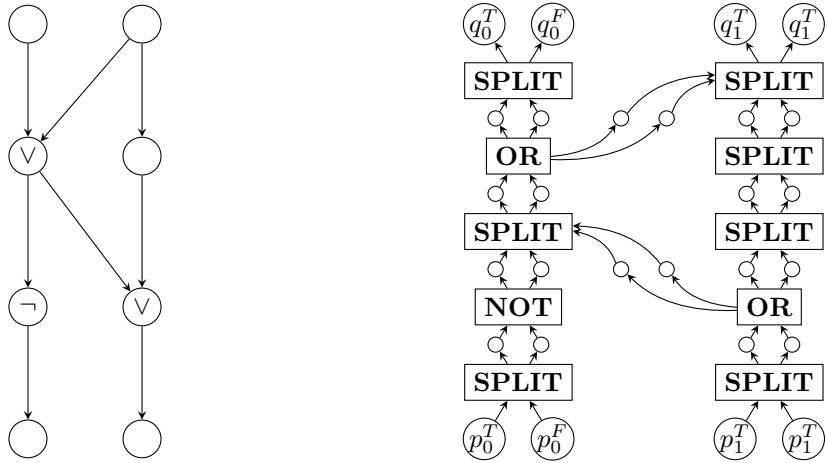
\begin{figure}[htp]%
    \begin{center}

      \begin{tikzpicture}
        [xscale=1.5,
        gate/.style={circle,draw,inner sep=0cm,minimum size=0.5cm},
        state/.style={circle,draw,inner sep=0cm,minimum size=0.2cm},
        wire/.style={draw,->,>=stealth},
        gadget/.style={rectangle,draw}]

        \node[gate] (in1) at (1,5.5) {};
        \node[gate] (in2) at (2,5.5) {};

        \node[gate] (g1) at (1,3.75) {$\vee$};
        \node[gate] (g2) at (2,3.75) {};
        \draw (in1) edge[wire] (g1);
        \draw (in2) edge[wire] (g1) edge[wire] (g2);

        \node[gate] (g4) at (1,1.75) {$\neg$};
        \node[gate] (g5) at (2,1.75) {$\vee$};
        \draw (g1) edge[wire] (g4) edge[wire] (g5);
        \draw (g2) edge[wire] (g5);

        \node[gate] (o1) at (1,0) {};
        \node[gate] (o2) at (2,0) {};
        \draw (g4) edge[wire] (o1);
        \draw (g5) edge[wire] (o2);

        \node[gate] (p1t) at (5,0) {$p_0^T$};
        \node[gate] (p1f) at (5.6,0) {$p_0^F$};
        \node[gate] (p2t) at (7.2,0) {$p_1^T$};
        \node[gate] (p2f) at (7.8,0) {$p_1^T$};

        \node[gadget] (s1) at (5.3,0.75) {\textbf{SPLIT}};
        \draw (p1t) edge[wire] (s1);
        \draw (p1f) edge[wire] (s1);
        \node[state] (s11) at (5.1,1.25) {};
        \node[state] (s12) at (5.5,1.25) {};
        \draw (s1) edge[wire] (s11) edge[wire] (s12);

        \node[gadget] (s2) at (7.5,0.75) {\textbf{SPLIT}};
        \draw (p2t) edge[wire] (s2);
        \draw (p2f) edge[wire] (s2);
        \node[state] (s21) at (7.3,1.25) {};
        \node[state] (s22) at (7.7,1.25) {};
        \draw (s2) edge[wire] (s21) edge[wire] (s22);

        \node[gadget] (n1) at (5.3,1.75) {\textbf{NOT}};
        \draw (s11) edge[wire] (n1);
        \draw (s12) edge[wire] (n1);
        \node[state] (n11) at (5.1,2.25) {};
        \node[state] (n12) at (5.5,2.25) {};
        \draw (n1) edge[wire] (n11) edge[wire] (n12);

        \node[gadget] (o1) at (7.5,1.75) {\textbf{OR}};
        \draw (s21) edge[wire] (o1);
        \draw (s22) edge[wire] (o1);
        \node[state] (o11) at (6.2,2.25) {};
        \node[state] (o12) at (6.6,2.25) {};
        \node[state] (o13) at (7.3,2.25) {};
        \node[state] (o14) at (7.7,2.25) {};
        \draw (o1) edge[wire,bend left] (o11) edge[wire,bend left] (o12) edge[wire] (o13) edge[wire] (o14);

        \node[gadget] (s3) at (5.3,2.75) {\textbf{SPLIT}};
        \draw (n11) edge[wire] (s3);
        \draw (n12) edge[wire] (s3);
        \draw (o11) edge[wire,bend right] (s3);
        \draw (o12) edge[wire,bend right] (s3);
        \node[state] (s31) at (5.1,3.25) {};
        \node[state] (s32) at (5.5,3.25) {};
        \draw (s3) edge[wire] (s31) edge[wire] (s32);

        \node[gadget] (s4) at (7.5,2.75) {\textbf{SPLIT}};
        \draw (o13) edge[wire] (s4);
        \draw (o14) edge[wire] (s4);
        \node[state] (s41) at (7.3,3.25) {};
        \node[state] (s42) at (7.7,3.25) {};
        \draw (s4) edge[wire] (s41) edge[wire] (s42);

        \node[gadget] (o2) at (5.3,3.75) {\textbf{OR}};
        \draw (s31) edge[wire] (o2);
        \draw (s32) edge[wire] (o2);
        \node[state] (o21) at (5.1,4.25) {};
        \node[state] (o22) at (5.5,4.25) {};
        \node[state] (o23) at (6.2,4.25) {};
        \node[state] (o24) at (6.6,4.25) {};
        \draw (o2) edge[wire] (o21) edge[wire] (o22) edge[wire,bend right] (o23) edge[wire,bend right] (o24);

        \node[gadget] (n2) at (7.5,3.75) {\textbf{SPLIT}};
        \draw (s41) edge[wire] (n2);
        \draw (s42) edge[wire] (n2);
        \node[state] (n21) at (7.3,4.25) {};
        \node[state] (n22) at (7.7,4.25) {};
        \draw (n2) edge[wire] (n21) edge[wire] (n22);

        \node[gadget] (s5) at (5.3,4.75) {\textbf{SPLIT}};
        \draw (o21) edge[wire] (s5);
        \draw (o22) edge[wire] (s5);

        \node[gadget] (s6) at (7.5,4.75) {\textbf{SPLIT}};
        \draw (o23) edge[wire,bend left] (s6);
        \draw (o24) edge[wire,bend left] (s6);
        \draw (n21) edge[wire] (s6);
        \draw (n22) edge[wire] (s6);

        \node[gate] (q1t) at (5,5.5) {$q_0^T$};
        \node[gate] (q1f) at (5.6,5.5) {$q_0^F$};
        \node[gate] (q2t) at (7.2,5.5) {$q_1^T$};
        \node[gate] (q2f) at (7.8,5.5) {$q_1^T$};
        \draw (s5) edge[wire] (q1t) edge[wire] (q1f);
        \draw (s6) edge[wire] (q2t) edge[wire] (q2f);

      \end{tikzpicture}
    \end{center}
    
    \caption{Example construction of $\mathcal{A}_C$, with $C$ on the left.}
    \label{fig:circuit-layers}%
  \end{figure}

  By the construction, every path from some output state $p^b_i \in Q$ to a given state $q \in Q$ has the same length in $\mathcal{A}_C$.
  For $\ell = 0, \ldots, 2 d(C) - 3$, let $S_\ell = \{ r_i^b(\ell) \in Q \;|\; 0 \leq i \leq n_\ell-1, b \in \{T,F\} \}$ be the states for which this length is $3 \ell$; they are exactly the input and/or output states of the gadgets used in the construction.
  We define consistent and excessive subsets of $S_\ell$ as above, and for $a \in \{T,F\}^{n_\ell}$, define $S_\ell(a) = \{ r_i^{a_i} \;|\; i = 0, \ldots, n_\ell-1 \}$.
  For $S \subseteq S_\ell$, Lemma~\ref{basicProperties}\ref{unionInsideEquiv} implies that $\delta_W(\{S\}, AAB)$ contains exactly those subsets of $S_{\ell+1}$ where each gate gadget is evaluated separately.
  Thus excessive subsets of $S_\ell$ are transformed into sets of excessive subsets of $S_{\ell+1}$, while consistent subsets $S_\ell(a)$ are transformed into sets of excessive sets and exactly the consistent sets $S_{\ell+1}(b)$ for those $b \in \{T,F\}^{n_{\ell+1}}$ that produce the vector $a$ when the corresponding gates are applied to them.
  We can then choose $p = 2 d(C) - 3$, and the claim follows.
\end{proof}

We also need the fact that the reversal of the winning set of a regular language can be recognized by a DFA of simply exponential size.
The construction is similar to the standard method of recognizing the reversed language of an alternating finite automaton~\cite{Chandra1981}.
The reversal of a language $L$ is denoted $L^R$.

\begin{definition}[Canonical DFA for $W(\mathcal{A})^R$]
  \label{def:rev-dfa}
  Let $\mathcal{A} = (Q, \{0,1\}, q_0, \delta, F)$ be a binary DFA.
  We construct a DFA $\mathcal{A}^W = (2^Q, \{A,B\}, q_0^W, \delta^W, F^W)$ as follows.
  The initial state is $q_0^W = F$, the final states are $F^W = \{S \subseteq Q \;|\; q_0 \in S\}$, and the transition function is
  \begin{align*}
    \delta^W(S, A) = {} & \{ q \in Q \;|\; \delta(q, 0) \in S \text{ or } \delta(q, 1) \in S \}, \\
    \delta^W(S, B) = {} & \{ q \in Q \;|\; \delta(q, 0) \in S \text{ and } \delta(q, 1) \in S \}.
  \end{align*}
\end{definition}

\begin{lemma}
  In the notation of Definition~\ref{def:rev-dfa}, we have $\mathcal{L}(\mathcal{A}^W) = W(\mathcal{L}(\mathcal{A}))^R$.
\end{lemma}

\begin{proof}
  Let $w \in \{A,B\}^*$.
  We prove by induction on $|w|$ that $\delta^W(F, w) = S := \{ q \in Q \;|\; w^R \in W(\mathcal{L}_q(\mathcal{A})) \}$, from which the claim follows.
  The case $w = \lambda$ is clear.
  Suppose then $w = v c$ with $c \in \{A,B\}$.
  Then $\delta^W(F, v) = \{ q \in Q \;|\; v^R \in W(\mathcal{L}_q(\mathcal{A})) \}$ by the induction hypothesis.
  If $c = A$, then
  \begin{align*}
    \delta^W(F, w) = {} & \delta^W(\delta^W(F, v), A) \\
    {} = {} & \{ q \in Q \;|\; v^R \in W(\mathcal{L}_{\delta(q, 0)}(\mathcal{A})) \cup W(\mathcal{L}_{\delta(q, 1)}(\mathcal{A})) \} \\
    {} = {} & \{ q \in Q \;|\; v^R \in A^{-1} W(\mathcal{L}_q(\mathcal{A})) \} = S.
  \end{align*}
  The analogous equations hold for $c = B$ with union replaced by intersection.
\end{proof}

\begin{theorem}
  Given a binary DFA $\mathcal{A}$ and a word $w \in \{A,B\}^*$,
  deciding whether $w \in W(\mathcal{L}(\mathcal{A}))$ is P-complete with respect to logarithmic space reductions.
\end{theorem}

\begin{proof}
  Let $C$ a boolean circuit with $k$ input nodes $x_1, \ldots x_k$ and one output node $y$, and let $a \in \{T,F\}^k$ be arbitrary.
  Let $\mathcal{A}_C$ be the DFA constructed in Lemma~\ref{circuitSimulation} for $C$ and $p$ the corresponding integer.
  We choose $q_0^T$ as the initial state and $F = \{ p_0^{a_0}, \ldots, p_{k-1}^{a_{k-1}} \}$ as the final states of $\mathcal{A}_C$.

  Now the game state $\delta_W(\{q_i^T\}, (AAB)^p)$ of the winning set automaton consists of excessive sets, which are never accepting, and the consistent sets $\{ p_0^{b_0}, \ldots, p_{k-1}^{b_{k-1}} \}$ for all $b \in \{T, F\}^k$ with $C(b) = T$.
  Since $F$ is the set of final state, $(AAB)^p \in W(\mathcal{A}_C)$ if and only if $C(a) = T$.
  The transformation of $C$ into $\mathcal{A}_C$ in Lemma~\ref{circuitSimulation} can be done in logarithmic space. 
  Thus we have a logspace-reduction of the circuit value problem into the winning set membership problem.

  To prove that the decision problem is in P, consider from the DFA for $W(\mathcal{L}(\mathcal{A}))^R$ given in Definition~\ref{def:rev-dfa}.
  It suffices to read $w^R$ in this automaton one symbol at a time, keeping track of a single set from $2^Q$, which can be done in polynomial time.
\end{proof}

\begin{theorem}
  Given a DFA $\mathcal{A}$ over $\{0,1\}$ and an NFA $\mathcal{B}$ over $\{A,B\}$, it is PSPACE-complete to decide whether $W(\mathcal{L}(\mathcal{A})) \cap \mathcal{L}(\mathcal{B}) \neq \emptyset$.
  The same holds if $\mathcal{B}$ is assumed deterministic.
\end{theorem}

\begin{proof}
  First we prove that the problem is in PSPACE.
  For this let $\mathcal{A}$ be a DFA with $n$ states and $\mathcal{B}$ an NFA with $k$ states.
  The DFA $\mathcal{A}^W$ of Definition~\ref{def:rev-dfa} for $W(\mathcal{L}(\mathcal{A}))^R$ has $2^n$ states, and reversing the transitions of $\mathcal{B}$ results in an NFA for $\mathcal{L}(\mathcal{B})^R$ with $k$ states.
  Their product NFA $\mathcal{P}$ recognizes the intersection using $k 2^n$ states, which can be computed as needed in polynomial space.
  Therefore by Savitch's accessibility theorem, we can decide if there is a path from the inital state of $\mathcal{P}$ to a final state in space $O(log^2(k 2^n)) = O( (n + log(k) )^2 )$.

  Conversely, we claim that given a DFA $\mathcal{A}$ and a word $w \in \{A,B\}^*$, deciding whether $W(\mathcal{L}(\mathcal{A})) \cap w^* \neq \emptyset$ is PSPACE-hard.
  Let $C$ a boolean circuit with $n$ gates, $k$ input nodes and $k$ output nodes $x_0, \ldots, x_{k-1}$, and take a word $a \in \{T,F\}^k$ and an index $i \in \{0, \ldots, k-1\}$.
  Create a new circuit $C'$ as follows.
  Add $k$ new output nodes $y_0, \ldots, y_{k-1}$.
  Add new gates so that $y_j = x_j \vee x_i$ for all $j \in \{0, \ldots, k-1\}$, so that the $x_j$ are no longer output nodes.
  Then for all $b \in \{T,F\}^k$, we have $C'(b) = C(b)$ if $C_i(b) = F$, and $C'(b) = T^k$ otherwise.
  
  Let the DFA $\mathcal{A}_{C'}$ and $p = O(n)$ be given by Lemma~\ref{circuitSimulation} for the circuit $C'$, and construct a new DFA $\mathcal{B}_{C'}$ from it as follows.
  Merge the two states $q_i^b$ and $p_i^b$ for each $i \in \{0, \ldots, k-1\}$ and $b \in \{T,F\}$.
  Add a new initial state $q_0$ and a gadget in the shape of a tree whose root is $q_0$ and whose leaves are the states $q_i^T$ similar to the left half of the \textbf{OR}-gate gadget in Figure~\ref{fig:example}, so that $\delta_W(\{\{q_0\}\}, (AAB)^{\lceil \log(k) \rceil}) = \{ \{ q_0^T, \ldots, q_{k-1}^T \} \}$.
  Choose $F = \{ q_0^{a_0}, \ldots, q_{k-1}^{a_{k-1}} \}$ for the final states of $\mathcal{B}$.
  For each $t \geq 0$ we now have $(AAB)^{\lceil \log(k) \rceil + t p} \in W(\mathcal{L}(\mathcal{B}_{C'}))$ if and only if $(C')^t(F^k) = T^k$, if and only if $C^s(F^k)_i = T$ for some $s \leq t$.
  Thus we have reduced the iterated circuit problem to the intersection problem with $w^*$, proving the PSPACE-hardness of the latter.
\end{proof}

Note that we cannot simply choose $F = \{q_i^T\}$ as the set of final states of $\mathcal{B}_{C'}$, since then the automaton would accept any game state with an excessive set that does not contain $q_i^F$, and such sets do not correspond to an iteration of the circuit.

\section{A Context-Free Language}

In this section we prove that the winning set operator does not in general preserve context-free languages by studying the winning set of the Dyck language $D \subseteq \{0,1\}^*$ of balanced parentheses. In our formalism, $0$ stands for an opening parenthesis and $1$ for a closing parenthesis.

\begin{proposition}
  \label{DyckWinningNotCFL}  
  The winning set of the Dyck language is not context-free.
\end{proposition}
\begin{proof}
  Take $L = W(D) \cap (AA)^*(BB)^*(AA)^*$, which is context-free if $W(D)$ is.
  We claim that
  $L =  \{ A^{2i} B^{2j} A^{2k} \;|\; i \geq j, k\geq 2j \}$.
  First, if Bob closes $2j$ parentheses, Alice must open at least
  $2j$ parentheses beforehand, so $i \geq j$ is necessary.
  If Bob opens $2j$ parentheses instead, when Alice plays a second time, she has to close $4j$ parentheses, hence $k \geq 2j$ is necessary.
  Thus the right hand side contains $L$.
  
  Conversely, Alice can win on $A^{2i} B^{2j} A^{2k}$ by leaving exactly $2j$ parentheses open before Bob's turns and then closing all open parentheses, so $L$ contains the right hand side.
  It's a standard exercise to prove that $L$ is not context-free.
\end{proof}

\bibliographystyle{plain}
\bibliography{references}

\begin{thebibliography}{10}

\bibitem{Anstee2002}
R.P. Anstee, Lajos Rónyai, and Attila Sali.
\newblock {Shattering News}.
\newblock {\em Graphs and Combinatorics}, 18(1):59--73, 2002.

\bibitem{complexity}
Sanjeev Arora and Boaz Barak.
\newblock {\em Complexity Theory: A Modern Approach}.
\newblock Cambridge University Press, New York, 2007.

\bibitem{auluck1951some}
FC~Auluck.
\newblock {On some new types of partitions associated with generalized Ferrers
  graphs}.
\newblock In {\em Mathematical Proceedings of the Cambridge Philosophical
  Society}, volume~47, pages 679--686. Cambridge University Press, 1951.

\bibitem{Chandra1981}
Ashok~K. Chandra, Dexter~C. Kozen, and Larry~J. Stockmeyer.
\newblock Alternation.
\newblock {\em Journal of the ACM}, 28(1):114--133, 1981.

\bibitem{goles2016pspace}
Eric Goles, Pedro Montealegre, Ville Salo, and Ilkka Törmä.
\newblock {PSPACE}-completeness of majority automata networks.
\newblock {\em Theoretical Computer Science}, 609:118--128, 2016.

\bibitem{boundedRegular}
Andrea Herrmann, Martin Kutrib, Andreas Malcher, and Matthias Wendlandt.
\newblock {Descriptional Complexity of Bounded Regular Languages}.
\newblock In Cezar C{\^a}mpeanu, Florin Manea, and Jeffrey Shallit, editors,
  {\em Descriptional Complexity of Formal Systems}, pages 138--152, Cham, 2016.
  Springer International Publishing.

\bibitem{HopcroftUllman}
John~E. Hopcroft, Rajeev Motwani, and Jeffrey~D. Ullman.
\newblock {\em Introduction to Automata Theory, Languages, and Computation (3rd
  Edition)}.
\newblock Addison-Wesley Longman Publishing Co., Inc., USA, 2006.

\bibitem{Friedl2003}
Friedl Katalin and Rónyai Lajos.
\newblock Order shattering and {W}ilson's theorem.
\newblock {\em Discrete Mathematics}, 270(1):127--136, 2003.

\bibitem{KleitmanMarkowskyDedekind75}
D.~Kleitman and G.~Markowsky.
\newblock {On Dedekind's Problem: The Number of Isotone Boolean Functions. II}.
\newblock {\em Transactions of the American Mathematical Society},
  213:373--390, 1975.

\bibitem{confVersion}
Pierre Marcus and Ilkka Törmä.
\newblock Descriptional complexity of winning sets of regular languages.
\newblock In Galina Jir{\'a}skov{\'a} and Giovanni Pighizzini, editors, {\em
  Descriptional Complexity of Formal Systems}, pages 130--141, Cham, 2020.
  Springer International Publishing.

\bibitem{PavlovSchraudner15}
Ronnie Pavlov and Michael Schraudner.
\newblock Classification of sofic projective subdynamics of multidimensional
  shifts of finite type.
\newblock {\em Transactions of the American Mathematical Society},
  367(5):3371--3421, 2015.

\bibitem{peltomaki2019winning}
Jarkko Peltomäki and Ville Salo.
\newblock On winning shifts of marked uniform substitutions.
\newblock {\em RAIRO-Theoretical Informatics and Applications}, 53(1-2):51--66,
  2019.

\bibitem{salo2014playing}
Ville Salo and Ilkka Törmä.
\newblock {Playing with Subshifts}.
\newblock {\em Fundamenta Informaticae}, 132(1):131--152, 2014.

\bibitem{Zielonka1998}
Wieslaw Zielonka.
\newblock Infinite games on finitely coloured graphs with applications to
  automata on infinite trees.
\newblock {\em Theoretical Computer Science}, 200(1):135--183, 1998.

\end{thebibliography}

\end{document}